\documentclass[11pt,a4paper]{article}
\usepackage{authblk}

\usepackage[T1]{fontenc}
\usepackage{lmodern}
\usepackage[margin=26mm]{geometry}
\usepackage{amsmath,amssymb,amsthm,mathtools}
\usepackage{microtype,enumitem,xcolor}
\usepackage{tikz}
\usetikzlibrary{arrows.meta,positioning,calc}
\usepackage{booktabs,tabularx,array}
\tikzset{
  paperbox/.style={draw=black!75,line width=.45pt,fill=black!2,
    align=center,inner xsep=6pt,inner ysep=5pt,font=\footnotesize,
    execute at begin node={\thinmuskip=3mu\relax\medmuskip=4mu\relax\thickmuskip=5mu\relax}},
  paperarrow/.style={-{Latex[length=2mm,width=1.3mm]},line width=.55pt},
  paperlabel/.style={font=\footnotesize,align=center,fill=white,inner sep=2pt}}

\usepackage{hyperref}

\hypersetup{hidelinks,
 pdftitle={Vanishing Ideals and the Computational Tractability of Sum-of-Squares over Boolean Domains},
 pdfauthor={Monaldo Mastrolilli}}
\newtheorem{theorem}{Theorem}[section]
\newtheorem{lemma}[theorem]{Lemma}
\newtheorem{proposition}[theorem]{Proposition}
\newtheorem{corollary}[theorem]{Corollary}
\theoremstyle{definition}
\newtheorem{definition}[theorem]{Definition}
\theoremstyle{remark}
\newtheorem{remark}[theorem]{Remark}
\newcommand{\R}{\mathbb R}
\newcommand{\Q}{\mathbb Q}
\newcommand{\N}{\mathbb N}
\newcommand{\I}{\mathcal I}
\newcommand{\J}{\mathcal J}
\newcommand{\B}{\mathcal B}
\newcommand{\Haff}{\mathcal H}
\newcommand{\ml}{\operatorname{ml}}

\newcommand{\aff}{\operatorname{aff}}
\newcommand{\conv}{\operatorname{conv}}
\newcommand{\rank}{\operatorname{rank}}
\newcommand{\vol}{\operatorname{vol}}
\newcommand{\OPT}{\operatorname{OPT}}
\newcommand{\cl}{\operatorname{cl}}
\newcommand{\spanQ}{\operatorname{span}_{\Q}}
\newcommand{\spanR}{\operatorname{span}_{\R}}
\newcommand{\bits}[1]{\langle #1\rangle}
\newcommand{\norm}[1]{\lVert #1\rVert}

\DeclareMathOperator{\dist}{dist}
\setlist[enumerate]{itemsep=4pt,topsep=5pt}
\title{Vanishing Ideals and the Computational Tractability of Sum-of-Squares over Boolean Domains
\thanks{Supported by the Swiss National Science Foundation projects
no.~200021\_207429/1 ``Ideal Membership Problems and the Bit Complexity
of Sum of Squares Proofs'' and no.~200021\_212929/1 ``Computational
methods for integrality gaps analysis''.}}

\author{Monaldo Mastrolilli}
\affil{SUPSI-IDSIA, Lugano, Switzerland\\
\texttt{monaldo.mastrolilli@supsi.ch}}

\date{}
\begin{document}
\maketitle

\begin{abstract}
Building on the bit-complexity framework of Raghavendra--Weitz and the
moment-SOS criteria of Gribling--Polak--Slot, we study the effective use
of truncated vanishing identities over Boolean polynomial systems.
A complete, constructible identity space gives an augmented moment SDP
that can be optimized with exact rational feasibility and arbitrary
additive accuracy. We give a self-contained geometric implementation:
an explicit affine reduction and a simplex of Boolean evaluations supply
the radius bounds required by rational ellipsoids. The identity space
can be constructed directly or extracted from a supplied graded
Gr\"obner basis, including one supplied at a higher truncation degree.
An explicit transfer theorem connects this augmented formulation to the
original system. Two-sided SoS derivations eliminate the added equality
axioms from certificates, with controlled degree and coefficient growth,
and imply containment of a projected higher-level moment relaxation in
the augmented body. Together with spectral coefficient bounds, this
gives polynomial-time search for rational proofs with an additive
perturbation.

For Min-closed linear systems, propagation constructs the identity space in
polynomial time at fixed degree and gives degree-$4t+4$ certificates for
both signs of each degree-$t$ basis element. Consequently, a rational degree-$(8d+4)$ proof of
$f+\varepsilon\ge0$ can be found in polynomial time for fixed $d$
whenever $f\ge0$ has a degree-$2d$ proof. The augmented degree-$2d$ moment SDP
can be optimized in polynomial time with exact rational feasibility and
a comparison to the original degree-$8d+4$ relaxation. Boolean complementation gives the same
results for Max-closed systems, including generalized packing and covering.
Min-closed systems thus provide a concrete application of the general
criteria. All complexity bounds are in the Turing model.
\end{abstract}
\clearpage
\tableofcontents
\clearpage

\section{Introduction}
\label{sec:intro}
The sum-of-squares (SoS) hierarchy is a central semidefinite-programming
method for polynomial optimization and discrete optimization.  Its
moment and localizing-matrix formulation goes back to
Lasserre~\cite{Lasserre2001} and Parrilo~\cite{Parrilo03}; over Boolean
variables it is one of the standard lift-and-project hierarchies, closely
related to other systematic relaxations for \(0/1\) programming
studied by Laurent~\cite{Laurent2003}.  At low levels, SoS captures many
useful valid inequalities and has underpinned influential approximation
algorithms, beginning with semidefinite methods such as
Goemans--Williamson~\cite{GoemansWilliamson1995}.

For a fixed degree, the corresponding semidefinite program has polynomial
dimension.  This fact alone does not prove polynomial-time solvability in
the Turing model.  Two obstructions must be addressed.  On the dual side,
a low-degree SoS certificate may require coefficients of enormous bit
length.  On the primal side, valid polynomial identities impose linear
relations among the moments.  The feasible moment set may consequently
lie in a proper affine subspace of the displayed coordinate space and have
empty interior there, even though it has a well-behaved relative interior.

\paragraph{Earlier criteria.}
O'Donnell~\cite{ODonnell} showed that low SoS degree alone does not control
coefficient size.  Raghavendra--Weitz~\cite{RW} control this obstruction by combining
spectral bounds with efficient Nullstellensatz derivations from the
equality axioms.  Their criterion shows that effective ideal membership
can turn the existence of a low-degree SoS proof into the existence of one
whose coefficients have bounded bit length.  Mastrolilli~\cite{Mastrolilli} (see also \cite{BharathiM22, BharathiM25})
subsequently used truncated Gr\"obner bases to construct effective
generating sets for Boolean constraint problems, thereby making the
algebraic input of the Raghavendra--Weitz criterion algorithmically
available and obtaining efficient formulations of the associated theta
body relaxations.  Bortolotti--Mastrolilli--Vargas~\cite{BMV} later
developed degree-automatability criteria based on algebraic derivations
over finite domains.

Building on Raghavendra--Weitz and classical results on rational
semidefinite optimization, Gribling--Polak--Slot (GPS)~\cite{Gribling23}
give algebraic and geometric conditions for polynomial-time computation
of moment-SOS bounds.  Their
algebraic condition identifies the matrix kernels forced by polynomial
identities, while their spectral condition supplies quantitative
nondegeneracy.  They also observe that a Gr\"obner basis can provide the
required algebraic information.  Their algebraic criterion provides the moment-SOS framework for our
primal result. We give a direct geometric implementation using the
classical rational ellipsoid method. We specialize this framework to rational Boolean systems for which the
complete truncated identity space can be constructed efficiently, give
explicit affine coordinates and geometric bounds, and connect these
identities to proof search over the original axioms through a
degree-controlled transfer. For Min- and Max-closed linear systems we
construct the identities and prove the required two-sided derivations.
The contributions are stated precisely below.

\paragraph{The common principle.}
Our starting point is the complete vector space of identities of bounded
degree that vanish on the Boolean feasible set.  This space has two
different algorithmic roles.  To search for proofs over the original
system, its identities must be \emph{derivable} from the original axioms
with controlled degree and coefficient size.  To optimize over moments,
the same identity space must be \emph{constructible}: a rational spanning
list must be available with controlled encoding length.  Derivability and
constructibility are logically distinct requirements, and the paper keeps
their consequences separate.

On the dual side, a two-sided SoS derivation of an identity $b=0$ means
proofs of both $b\ge0$ and $-b\ge0$.  Such derivations allow the identity
to be used and later eliminated from an SoS certificate.  On the primal
side, imposing every truncated identity removes precisely the affine
directions that vanish on all feasible Boolean points.  Parametrizing the
resulting affine hull by independent coordinates exposes a genuine inner
ball.  Thus the same identity space links proof search and moment
optimization, but through two different kinds of algorithmic access.

Figure~\ref{fig:unified} displays these two branches.  The two ways of
constructing the primal input are developed in
Section~\ref{sec:primal-routes}.

\begin{figure}[htbp]
\centering
\begin{tikzpicture}
\node[paperbox,text width=8.3cm] (ideal) at (0,0)
  {Truncated vanishing identities\\$\I(S)\cap\Q[x]_{\le2d}$};
\node[paperbox,text width=6.0cm,minimum height=1.45cm,anchor=north]
  (derive) at ($(ideal.south)+(-3.65,-.85)$)
  {Two-sided SoS proofs of $\pm b\ge0$\\
   for $b\in\B_{2d}$; controlled degree\\and coefficient bounds};
\node[paperbox,text width=6.0cm,minimum height=1.45cm,anchor=north]
  (construct) at ($(ideal.south)+(3.65,-.85)$)
  {Construct a complete rational list\\
   Direct identities or graded $\B_D$\\at algebraic degree $D\ge2d$};
\node[paperbox,text width=6.0cm,minimum height=1.15cm,below=8mm of derive]
  (bounded)
  {Bounded certificates\\Controlled degree increase};
\node[paperbox,text width=6.0cm,minimum height=1.15cm,below=8mm of construct]
  (geometry)
  {Affine hull of moments\\Inner / outer radii and exact separation};
\node[paperbox,text width=6.0cm,minimum height=1.25cm,below=8mm of bounded]
  (dual)
  {\textbf{Dual:} rational proof search\\for $f+\varepsilon\ge0$};
\node[paperbox,text width=6.0cm,minimum height=1.25cm,below=8mm of geometry]
  (primal)
  {\textbf{Primal:} rational $y^*\in K_{2d}$\\
   Exact feasibility and $\varepsilon$-optimality};
\draw[paperarrow] (ideal.south) -- ++(0,-.4) -| (derive.north);
\draw[paperarrow] (ideal.south) -- ++(0,-.4) -| (construct.north);
\draw[paperarrow] (derive) -- (bounded);
\draw[paperarrow] (bounded) -- (dual);
\draw[paperarrow] (construct) -- (geometry);
\draw[paperarrow] (geometry) -- (primal);
\end{tikzpicture}
\caption{The common algebraic object and its two algorithmic uses.
The left branch assumes bounded two-sided certificates; the right branch
assumes an effective description of all truncated identities. Two-sided
derivations also yield the hierarchy comparison in
Figure~\ref{fig:primal-roles}.}
\label{fig:unified}
\end{figure}
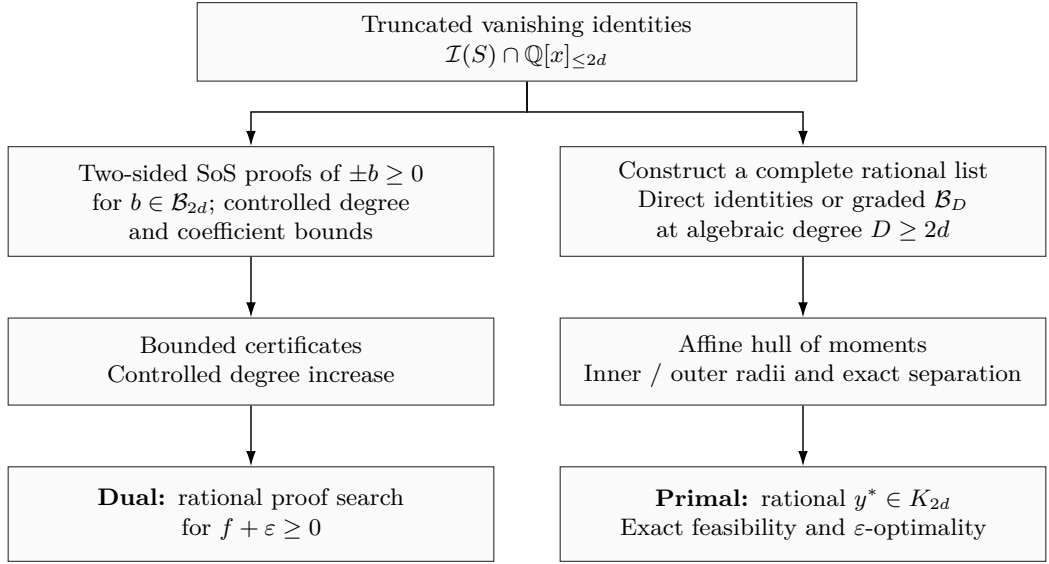

\paragraph{Contributions.}
The paper makes the preceding principle quantitative and algorithmic in
four steps.  We use ordinary degree throughout: a degree-$2d$
pseudoexpectation has moments through degree $2d$ and a moment matrix
indexed by monomials through degree $d$.

\emph{1. A dual criterion and an explicit transfer theorem.}
Theorem~\ref{thm:dual} assumes that both signs of every element of a
truncated reduced graded Gr\"obner basis have degree-$O(d)$ SoS proofs
with bounded coefficients.  Under this assumption, every polynomial that
has a degree-$2d$ SoS proof also has a degree-$O(d)$ proof with controlled
coefficient size.  The bounded-certificate search reduction in
Appendix~\ref{app:search}, using standard convex optimization and an
explicit Boolean residual correction, then finds a rational Gram certificate of
$f+\varepsilon\ge0$ in polynomial time whenever $f\ge0$ has
a degree-$2d$ proof.

The spectral estimate follows Raghavendra--Weitz, with the rational
moment and localizing-matrix bounds of Gribling--Polak--Slot.
Theorem~\ref{thm:transfer} quantifies the standard algebraic substitution
needed when the useful identities are not original equality axioms.  It shows how
to eliminate every added equation $b=0$ by substituting two-sided SoS
proofs of $b\ge0$ and $-b\ge0$.  The theorem tracks the resulting degree
and coefficient growth: if $b$ has degree $t$ and its two signs have
degree-$D_t$ proofs, then its use in a degree-$2d$ certificate can be
eliminated within degree $D_t+2(2d-t)$.  This permits identities derived
using inequalities to play the role that equality consequences play in
the Raghavendra--Weitz criterion.

\emph{2. A constructive Boolean realization of the primal criterion.}
Theorem~\ref{thm:main} starts from a rational, polynomial-size spanning
list of all identities of degree at most $2d$.  Adding their moment
equalities to the original degree-$2d$ relaxation gives the augmented body
$K_{2d}$.  We prove directly that these equalities describe its complete
affine hull, compute independent rational coordinates, construct an inner
ball of inverse-exponential radius and a polynomially bounded outer ball,
and give a strong rational separation oracle.  The rational ellipsoid
method then returns a rational point of $K_{2d}$ that satisfies every
constraint exactly and is optimal within any prescribed additive error.

This optimization principle is a Boolean instance of
Gribling--Polak--Slot~\cite[Proposition~6 and Theorem~7]{Gribling23}:
the complete identities supply their algebraic hypothesis, and Boolean
evaluation moments supply their spectral hypothesis. Their SDP algorithm
already has an exactly feasible rational output
\cite[Theorem~15 and Section~3.4]{Gribling23}. Here we give an explicit,
self-contained construction of the affine reduction and quantitative
ellipsoid data from the supplied identity list. This construction, and
its combination with degree-controlled derivability, are the parts
developed in this paper.

\emph{3. Two effective constructions of the same primal body.}
The input to the primal theorem can be obtained in two ways.  The direct
route constructs any rational spanning list of the truncated identity
space, for example from closure information.  The Gr\"obner-basis route starts from a supplied or effectively computed reduced graded basis $\B_D$ of the vanishing ideal, where
$D\ge 2d$; the case $D=2d$ is sufficient, while a basis supplied at
a higher degree can simply be truncated to the required degree.  Graded
division extracts all identities needed at degree $2d$, while standard
monomials provide the independent moment coordinates.  Theorem
\ref{thm:supplied-route} records the running time and encoding length of
this reduction.  Both routes impose the same affine equations and hence
optimize the same body $K_{2d}$; they are alternative implementations of
one geometric idea, not different relaxations.  The use of truncated
Gr\"obner bases as effective generators connected with the
Raghavendra--Weitz criterion, and for the efficient formulation of theta
bodies, originates in Mastrolilli~\cite{Mastrolilli}.  Our supplied-basis
theorem integrates that algebraic route with the present exact-feasible
moment optimization result.

\emph{4. Connection with the original moment hierarchy.}
Constructibility of the identities suffices to optimize $K_{2d}$.
Derivability supplies an additional conclusion.  If every identity used
in $K_{2d}$ has two-sided SoS derivations of degree at most the even integer
$\kappa\ge2d$, then soundness gives
\[
 \pi_{2d}(S_\kappa)\subseteq K_{2d}\subseteq S_{2d}.
\]
Here $S_{2d}$ is the original moment relaxation with moments through degree
$2d$, and $\pi_{2d}$ discards moments of degree above $2d$.  Consequently,
for a maximization problem the pseudoexpectation produced over $K_{2d}$
has objective value at least the optimum of the original degree-$\kappa$
relaxation, up to the requested additive error.  The moment degree of the
algorithm remains $2d$; $\kappa$ records the degree needed to derive and
eliminate the added identities.

\paragraph{Min-closed systems as an application.}
Min-closed systems supply both the algebraic and algorithmic hypotheses
of the general criteria explicitly. The same closure structure supports
the two algorithms.
A Boolean linear inequality is \emph{Min-closed} if the coordinatewise
minimum of two satisfying Boolean vectors also satisfies it. We call a
system of such rows a generalized packing system. Classical packing
inequalities $\sum_i a_i x_i\le b$, $a_i\ge0$, are examples, but Min-closed
rows can have mixed signs and need not define a downward-closed family.
The example $x_2\le x_1$ illustrates this difference. Generalized covering
is defined by Max-closed rows, including classical nonnegative covering
inequalities. The assumption is on each given row, not only on the
intersection of the rows.

For Min-closed systems, closures of small supports construct all truncated
identities without enumerating the feasible set.  This verifies the
constructibility hypothesis and yields the primal algorithm.  Separately,
Boolean propagation admits degree-four SoS certificates and refutes an
infeasible Min-closed system in degree four.  Conditioning these
refutations proves both signs of every degree-$t$ basis element in degree
$4t+4$, with polynomial coefficient bit length.  This verifies the
derivability hypothesis and yields the dual algorithm and the comparison
with a higher level of the original moment hierarchy.

For a moment vector $u$, let $L_u$ denote the associated linear
functional, defined on monomials by $L_u(x^\alpha)=u_\alpha$ and
extended linearly to polynomials.

\begin{theorem}[Algorithmic consequences for generalized packing and covering]
\label{thm:intro}
Fix $d\ge1$. For a rational Boolean linear system whose rows are all
Min-closed, or whose rows are all Max-closed, the following tasks have
polynomial bit complexity in the input length and the number of accuracy
bits.
\begin{enumerate}[label=\textup{(\roman*)}]
\item Under the promise that a given rational $f\ge0$ has a degree-$2d$
SoS proof, find a rational Gram proof of $f+\varepsilon\ge0$ of degree
at most $8d+4$.
\item Either report that the Boolean system is infeasible, or return a
rational degree-$2d$ pseudoexpectation that is exactly feasible and
$\varepsilon$-optimal for $K_{2d}$. Its objective value is at least
$\max_{u\in S_{8d+4}}L_u(f)-\varepsilon$.
\end{enumerate}
In \textup{(ii)} the degree of $f$ is at most $2d$ and maximization is used.
The precision may be specified as $\varepsilon=2^{-s}$, with time polynomial
in $s+1$. The augmented SDP uses only moments through degree $2d$.
\end{theorem}

Here $L_u$ is the linear functional associated with the moment vector
$u$: it assigns $u_T$ to the squarefree monomial $x_T$ and is extended
to polynomials by linearity and the Boolean equations. The notation
$S_t$ denotes the original degree-$t$ moment relaxation, while $K_{2d}$
also imposes all vanishing identities through degree $2d$.
Theorem~\ref{thm:intro} combines Corollaries~\ref{cor:min-dual} and~\ref{cor:min-primal}
with Proposition~\ref{prop:complement}, giving polynomial-time algorithms
on both the certificate and moment sides for these systems.

The Min-closed application therefore contributes the missing bridge from
the abstract criteria to compact systems of rational linear inequalities.
It gives, for the same original formulation, polynomial-time proof search
in degree $8d+4$ and polynomial-time optimization of the augmented
degree-$2d$ moment body, together with
$\pi_{2d}(S_{8d+4})\subseteq K_{2d}$.  Max-closed systems follow by
Boolean complementation.  The geometric and optimization arguments are
included in full, with the rational ellipsoid theorem of
Gr\"otschel--Lov\'asz--Schrijver~\cite{GLS} as the algorithmic tool.

Huber and Mastrolilli~\cite{HuberMastrolilli2026} apply
Theorem~\ref{thm:intro} to the approximation of Chv\'atal--Gomory closures.
They obtain a polynomial-time approximation scheme for the first
Chv\'atal--Gomory closure generated by multipliers in
$\{0\}\cup[1/f,1]$, for fixed $f\ge2$, over capacity-bounded
Min-closed systems. The primal criterion developed here supplies the
algorithmic ingredient: polynomial-time computation of rational
pseudoexpectations satisfying the required bounded-degree SoS
constraints exactly, with controlled additive error in the objective.

A further application is given by Benedetto Proen\c{c}a, Encz, and
Mastrolilli~\cite{ProencaEnczMastrolilli2026}, who construct truncated
vanishing identities and their two-sided SoS derivations for soft
Majority-closed Boolean linear systems. Applying the criteria developed
here yields polynomial-time proof search with controlled degree increase
and exact-feasible rational optimization of the augmented moment
relaxations at fixed degree, providing the computational foundation
for their approximation algorithm.

\paragraph{Organization.}
Section~\ref{sec:prelim} fixes the degree and encoding conventions.
Sections~\ref{sec:dual} and~\ref{sec:primal} prove the two general criteria.
Section~\ref{sec:min} applies the criteria to Min-closed systems,
compares the two primal constructions, and transfers the conclusions
to Max-closed systems by Boolean complementation in
Section~\ref{sec:complement}. Its final subsection discusses generalized
packing and covering.
After the references, Appendices~\ref{app:algebra}--\ref{app:search}
collect the technical tools used in the main text. They recall standard
facts (see e.g. \cite{GLS}) from rational linear algebra, positive-semidefinite separation,
the rational ellipsoid method, and perturbation-and-repair arguments,
while proving the quantitative affine-hull, radius, reconstruction, and
encoding estimates required by our applications.

\section{Preliminaries and conventions}
\label{sec:prelim}
Let $n\ge1$ and consider a rational polynomial system
\begin{equation}\label{eq:system}
 g_i(x)\ge0\quad(i\in[m]),\qquad
 h_j(x)=0\quad(j\in[\ell]),
\end{equation}
including the Boolean equations $x_a^2-x_a=0$ for all $a\in[n]$.
Its real feasible set is $S\subseteq\{0,1\}^n$, and
\[
 \I(S)=\{q\in\R[x_1,\ldots,x_n]:q(s)=0\text{ for every }s\in S\}.
\]
Fix $d\ge1$. Write $\J_t=\{T\subseteq[n]:|T|\le t\}$ and
\[
 N=|\J_{2d}|,\qquad W=\binom{n+2d}{2d},\qquad N\le W\le(n+1)^{2d}.
\]
Here $2d$ is the maximum moment degree and $d$ is the moment-matrix order.
Let $\mathsf{len}$ be the total binary encoding length of the input system, including
dimensions and monomial supports. Rational numbers have binary numerators
and positive denominators; $\bits u$ denotes total encoding length.
All computational bounds count bit operations. All logarithms are to
base two, and $B(a,\rho)$ denotes a closed Euclidean ball.
Accuracy parameters $\varepsilon$ are positive rationals. For a polynomial
$p=\sum_\alpha p_\alpha x^\alpha$, write
$\norm p_1=\sum_\alpha|p_\alpha|$. For real symmetric matrices, use the
trace inner product $\langle A,B\rangle=\operatorname{tr}(AB)$.

\paragraph{Moment and proof conventions.}
For $T\subseteq[n]$, write $x_T=\prod_{a\in T}x_a$ and
$x_\varnothing=1$. Multilinear reduction $\ml$ replaces each positive
exponent by one. For even $t\ge2$, a vector $y\in\R^{\J_t}$ defines
\[
 L_y(p)=\sum_{T\in\J_t}p_Ty_T,
 \qquad \ml(p)=\sum_{T\in\J_t}p_Tx_T,\qquad \deg p\le t.
\]
The original pseudoexpectation relaxation $S_t$ consists of the vectors
satisfying
\begin{align}
 y_\varnothing&=1,\qquad
 M_{t/2}(y)=(y_{A\cup B})_{A,B\in\J_{t/2}}\succeq0,
 \label{eq:moment}\\
 M_{a_i}(g_i y)&=(L_y(g_i x_Ax_B))_{A,B\in\J_{a_i}}\succeq0,
 \quad a_i=\left\lfloor\frac{t-\deg g_i}{2}\right\rfloor,
 \quad \deg g_i\le t,\label{eq:local}\\
 L_y(x^\alpha h_j)&=0
 \quad(\alpha\in\N^n,\ |\alpha|+\deg h_j\le t).
 \label{eq:equalities}
\end{align}
Zero input polynomials are discarded. A constraint of degree greater
than $t$ has no block or multiplier at this truncation, but still defines
$S$. Degrees refer to ordinary total degree before multilinearization.

A degree-$t$ SoS proof of $p\ge0$ is an identity
\begin{equation}\label{eq:sos}
 p=\sigma_0+\sum_i\sigma_i g_i+\sum_j u_jh_j,
 \qquad
 \deg\sigma_0,\ \deg(\sigma_i g_i),\ \deg(u_jh_j)\le t,
\end{equation}
where every $\sigma_i$ is a sum of squares and the $u_j$ are arbitrary
polynomials. Real coefficients suffice for soundness: if $y\in S_t$,
then $L_y(p)\ge0$, by applying $L_y$ to the squares, localizing squares,
and equality multiples in~\eqref{eq:sos}. Boolean reductions preserve
degree: repeatedly use
$x_a^v-x_a^{v-1}=x_a^{v-2}(x_a^2-x_a)$ for $v\ge2$, with the other
monomial factors retained. Thus polynomial identities proved modulo the
Boolean ideal can be lifted to this convention without increasing degree.

For even $\kappa\ge2d$, let $\pi_{2d}$ discard all moments of degree
greater than $2d$. It maps $S_\kappa$ into $S_{2d}$: every lower-degree
PSD matrix is a principal submatrix of the corresponding higher-degree
matrix, and every required equality is already imposed at degree $\kappa$.

\paragraph{Gram encoding and certificate size.}
Write $v_s$ for the vector of all ordinary monomials of degree at most $s$.
A sum of squares $\sigma$ of degree at most $2s$ is represented by
$\sigma=v_s^\top Qv_s$, where $Q\succeq0$ is symmetric.
We encode a rational certificate by these Gram matrices and the coefficients
of the equality multipliers. A size bound controls the magnitudes of these
entries, which may be real. A rational PSD
Gram matrix is an exact rational certificate even if a chosen square-root
factorization has irrational entries. Degree automatability below means
polynomial-time search for a rational certificate of $f+\varepsilon\ge0$,
under the promise that $f\ge0$ has a degree-$2d$ certificate, allowing degree
$O(d)$ and time polynomial in the number of monomials at that degree and
the requested number of accuracy bits.

\paragraph{Vanishing identities and graded bases.}
The ideal $\I(S)$ is considered over $\R$; its degree-bounded coefficient
spaces have rational bases, since they are kernels of Boolean evaluation
matrices. We always use a fixed graded lexicographic order. The reduced
Gr\"obner basis is monic; $\B_r$ denotes its elements of ordinary degree at
most $r$. An identity list is \emph{complete at degree $r$} if its rational
span is $\I(S)\cap\Q[x]_{\le r}$. A truncated basis is generally not such a
linear spanning list until its admissible monomial multiples are included.
Appendix~\ref{app:algebra} proves the standard division and Boolean reduction
facts used for this purpose.

\begin{definition}[Strong rational separation]\label{def:strong}
For a closed convex $C\subseteq\R^p$, a strong rational separation oracle
accepts a rational query $z$ only if $z\in C$; otherwise it returns
$a\in\Q^p\setminus\{0\}$ and $b\in\Q$ satisfying
$a^\top x\le b<a^\top z$ for every $x\in C$.
Its time and output length are polynomial in the oracle-instance encoding
and the query encoding. This is exact separation, including on the boundary.
\end{definition}

\section{The dual criterion: controlling certificates through vanishing identities}
\label{sec:dual}
We use the spectral coefficient-control method of
Raghavendra--Weitz~\cite{RW}, closely related to the moment-SOS
certificate bounds in~\cite[Theorems~9 and~22]{Gribling23}.
The first lemma gives the bounded representation modulo the vanishing
ideal. Theorem~\ref{thm:transfer} next shows how to eliminate added identities,
with degree and coefficient bounds. Combining these two results gives
the dual criterion.

A degree-\(2d\) SoS proof may contain Gram-matrix components whose associated
polynomials vanish at every feasible point. Such components do not affect
the values of the proof on the feasible set, but they can carry arbitrarily
large coefficients. We project them away and collect the resulting
difference in a polynomial of the vanishing ideal. This projection argument
is used only to prove the existence of an equivalent certificate whose
coefficients are bounded by a polynomially encoded number \(B\). The algorithm
does not construct the projected Gram matrices, nor does it enumerate the
feasible set \(S\). Once the bound \(B\) is known, it searches directly, in
the original Gram encoding, for any rational certificate satisfying the
coefficient bound \(B\).

The following lemma adapts the Gram-projection step in the proof of
Raghavendra--Weitz~\cite[Theorem~10]{RW} to the present Boolean
setting. We retain the resulting vanishing residual as an element of
the truncated vanishing ideal; it will be eliminated later through
two-sided SoS derivations.
For the arithmetic spectral bounds we use the integer-matrix argument
of~\cite[Lemmas~6 and~7]{RW} and its moment and localizing-matrix
formulations in~\cite[Theorem~7 and Propositions~20 and~21]{Gribling23}.
These bounds apply to rational Boolean evaluation moments, including
when some inequalities vanish at feasible points.

\begin{lemma}[Bounded certificate modulo the vanishing ideal]
\label{lem:bounded-modulo}
Suppose $S\ne\varnothing$ and a rational polynomial $f$ has a degree-$2d$
SoS proof. Put $\mathsf{len}_f=\mathsf{len}+\bits f$. There exist PSD Gram matrices $Q_i'$ and
$q\in\I(S)\cap\R[x]_{\le2d}$ such that
\[
 f=\sum_{i=0}^m g_i v_{a_i}^{\top}Q_i'v_{a_i}+q,
 \qquad g_0=1,\quad a_i=\left\lfloor(2d-\deg g_i)/2\right\rfloor,
\]
where negative $a_i$ are omitted. Every matrix entry and every coefficient
of $q$ has magnitude at most $2^{\operatorname{poly}(W,\mathsf{len}_f)}$, for a
computable uniform polynomial. The coefficients here may be real.
\end{lemma}
\begin{proof}
Start with any Gram proof, without a coefficient bound. For each retained
block form, for the proof only,
\[
 M_i=\frac1{|S|}\sum_{s\in S}g_i(s)v_{a_i}(s)v_{a_i}(s)^\top\succeq0.
\]
Let $\Pi_i$ be the orthogonal projection onto its range and put
$Q_i'=\Pi_i Q_i\Pi_i$. A vector $u$ in $\ker M_i$ satisfies
$\sum_{s\in S}g_i(s)(u^\top v_{a_i}(s))^2=0$. Every summand is nonnegative;
thus $u^\top v_{a_i}(s)=0$ whenever $g_i(s)>0$. Consequently
$g_i(s)v_{a_i}(s)^\top Q_i'v_{a_i}(s)$ equals the original summand at
every $s\in S$. The equality terms in the original proof vanish there,
so the displayed residual $q$ vanishes there as well.

Choose a positive common denominator $b_i$ of the coefficients of $g_i$,
with $\log b_i\le \mathsf{len}$. The matrix
$A_i=|S|b_iM_i$ is an integer PSD matrix of order at most $W$.
Its entries have magnitude at most
$U_0=2^{n+2\mathsf{len}}(\mathsf{len}+1)$: ordinary monomials evaluate to zero or one, there
are at most $\mathsf{len}$ input terms, and $|S|\le2^n$. By
Lemma~\ref{lem:integer-spectrum}, every positive eigenvalue of $A_i$
is at least $(WU_0)^{-W}$. Therefore every positive eigenvalue of $M_i$
is at least
\[
 \delta=\frac{(WU_0)^{-W}}{2^{n+\mathsf{len}}}>0.
\]
If $M_i=0$, take $Q_i'=0$. Otherwise $M_i\succeq\delta\Pi_i$.
Averaging the original proof on $S$ gives
\[
 0\le\langle M_i,Q_i'\rangle
 \le \frac1{|S|}\sum_{s\in S} f(s)
 \le W2^{\bits f}.
\]
Here every other averaged Gram term is nonnegative. Since
$Q_i'=\Pi_iQ_i'\Pi_i\succeq0$, the left side is at least
$\delta\operatorname{tr}(Q_i')$. Hence every entry of $Q_i'$ has
magnitude at most $W2^{\bits f}/\delta$; for a PSD matrix this follows
from $|Q_{uv}|\le\sqrt{Q_{uu}Q_{vv}}\le\operatorname{tr}Q$.
Expanding at most $(m+1)W^2$ Gram terms and the input coefficients now
bounds $\norm q_1$ by $2^{\operatorname{poly}(W,\mathsf{len}_f)}$.
All the displayed majorants are computable from the encoded input, even
though the projections and the set $S$ need not be computed.
\end{proof}

The following theorem allows identities with two-sided SoS proofs
from the original system to be used as additional equality axioms and
then eliminated from a certificate. The final certificate proves the
same inequality from the original axioms alone, with controlled degree
and coefficient growth. On the moment side, this gives containment of
a projected higher-degree relaxation in the augmented lower-degree body.
The substitution is classical: it shows that the support $M\cap(-M)$
of a quadratic module is an ideal
\cite[Section~4.1, Eq.~(16)]{KlepSchweighofer12}.

\begin{theorem}[Transfer from added identities to the original system]
\label{thm:transfer}
Let $r\ge2$ be even, and let $\mathcal E$ be a finite list of nonzero
polynomials of degree at most $r$. Suppose both signs of every
$b\in\mathcal E$ have SoS proofs from~\eqref{eq:system} of degree
$D_b$. Define the even degree
\[
 \Delta_{\mathcal E}
 =2\left\lceil\frac12\max\left\{r,
       \max_{b\in\mathcal E}[D_b+2(r-\deg b)]\right\}\right\rceil,
\]
omitting the inner maximum if $\mathcal E$ is empty. Then:
\begin{enumerate}[label=\textup{(\roman*)}]
\item Every degree-$r$ SoS certificate of $f\ge0$ from the original
system augmented by the equality axioms $b=0$, $b\in\mathcal E$,
can be transformed into an SoS certificate of the same inequality
$f\ge0$ from the original system alone, of degree at most
$\Delta_{\mathcal E}$.
If all entries of the augmented certificate, including its equality
multipliers, have magnitude at most $A\ge1$, and the two-sided
certificates have entries of magnitude at most $B\ge1$, the transferred
certificate has entries of magnitude at most
\[
 A+2|\mathcal E|W_{\Delta_{\mathcal E}}^2(A+1)^2B,
 \qquad W_{\Delta_{\mathcal E}}=\binom{n+\Delta_{\mathcal E}}{
                                      \Delta_{\mathcal E}}.
\]
Rational input certificates give a rational output by an explicit
polynomial-bit-time substitution in their encoded sizes and
$W_{\Delta_{\mathcal E}}$.
\item Let $S_r^{\mathcal E}$ be the degree-$r$ moment relaxation
with these equality axioms and all degree-allowed multipliers. Then
\[
 \pi_r(S_{\Delta_{\mathcal E}})\subseteq S_r^{\mathcal E}.
\]
If the admissible multiples of $\mathcal E$ span
$\I(S)\cap\R[x]_{\le r}$, then $S_r^{\mathcal E}=K_r$, the moment
body obtained by imposing all vanishing identities through degree $r$.
The inclusion requires only the two-sided degree bounds.
\end{enumerate}
\end{theorem}
\begin{proof}
Write the augmented certificate as
\[
 f=\sum_{i=0}^m\sigma_i g_i+\sum_j u_jh_j+
                  \sum_{b\in\mathcal E}a_b b,
 \qquad \deg a_b\le r-\deg b.
\]
For every multiplier $a=a_b$, use the polynomial identity
\begin{equation}\label{eq:signed-product}
 ab=\left(\frac{a+1}{2}\right)^2b
       +\left(\frac{a-1}{2}\right)^2(-b).
\end{equation}
Substituting the two assumed certificates and multiplying by the
displayed squares gives a certificate for $ab$ of degree at most
$D_b+2\deg a_b$. Summing proves the degree assertion in~\textup{(i)}.

Multiplication of a Gram polynomial by $u^2$ replaces its matrix $Q$
by $T_uQT_u^\top$, where $T_u$ is the coefficient matrix for
multiplication by $u$. Here every coefficient of
$u=(a_b\pm1)/2$ has magnitude at most $A+1$.
Each new Gram entry is therefore a sum of at most
$W_{\Delta_{\mathcal E}}^2$ terms of magnitude at most $(A+1)^2B$.
The same bound holds for coefficients of the multiplied equality
terms, by expanding $u^2u_j$. There are two substituted certificates
per $b$, in addition to the original terms. This proves the displayed
magnitude bound. With rational certificates, the indicated additions
and multiplications are polynomially many operations on rationals of
polynomial encoding length, which proves the bit-complexity statement.

For~\textup{(ii)}, take $v\in S_{\Delta_{\mathcal E}}$ and any
polynomial $a$ with $\deg a+\deg b\le r$. The same substitution gives
proofs of both $ab\ge0$ and $-ab\ge0$ within degree
$\Delta_{\mathcal E}$. Soundness yields $L_v(ab)=0$.
Restriction to degree $r$ preserves this identity and the original
moment constraints, proving the inclusion. If the admissible
multiples span the full truncated identity space, their moment
equalities are exactly those defining $K_r$.
\end{proof}

The transfer is an algebraic substitution, while the quantitative
coefficient bound makes it usable in the Turing model. Applied after
the spectral lemma and graded division, it proves the dual criterion.
Applied to the moment functionals alone, it gives the hierarchy
comparison below. Effective construction of the identity space supplies
the separate algorithm for optimizing the augmented body.

The dual criterion should be compared with the equality-based
Nullstellensatz condition in~\cite[Theorem~10]{RW} and the
polynomial-calculus criterion in~\cite[Theorem~2.1]{BMV}.
Berkholz's simulation gives an SoS derivation of $-b^2\ge0$ from a
degree-$D$ polynomial-calculus derivation of $b=0$, within degree $2D$
\cite[Lemma~3.1]{Berkholz18}. The hypothesis below instead asks directly
for certificates of both $b\ge0$ and $-b\ge0$, with controlled degree
and coefficient bounds; inequalities may participate in these
certificates.

\begin{theorem}[Dual criterion]
\label{thm:dual}
Let the rational Boolean system be as in Section~\ref{sec:prelim}.
Let $\B_{2d}$ be its reduced graded basis truncated at degree $2d$.
Assume its coefficients have magnitude at most $C\ge1$ and, for each
$b\in\B_{2d}$, both $b\ge0$ and $-b\ge0$ have Gram certificates with
entries and equality coefficients of magnitude at most $B\ge1$, of degree
$D_b\le c d$, for a fixed constant $c$. Put
\begin{equation}\label{eq:dual-degree}
 \Delta=2\left\lceil\frac12\max\left\{2d,
       \max_{b\in\B_{2d}}[D_b+2(2d-\deg b)]\right\}\right\rceil.
\end{equation}
The maximum over an empty basis is omitted. Every degree-$2d$ SoS proof of
a rational $f\ge0$ can be replaced by a degree-$\Delta$ proof whose Gram entries
and equality coefficients have magnitude at most
$2^{\operatorname{poly}(W,\mathsf{len}+\bits f,\log C,\log B)}$.
Given computable numerical bounds $B,C$ and an even search degree $O(d)$ bounding
the right side of~\eqref{eq:dual-degree}, one can search at that degree for a rational
Gram proof of $f+\varepsilon\ge0$ in polynomial time in these parameters
and $\bits\varepsilon$, under the promise that a degree-$2d$ proof of
$f\ge0$ exists. Neither the basis nor its two-sided certificates must be
supplied to the search algorithm.
\end{theorem}
\begin{proof}
For nonempty $S$, apply Lemma~\ref{lem:bounded-modulo} and divide its
residual by the monic graded basis. Lemma~\ref{lem:division}, valid also
for real coefficients, gives
\[
 q=\sum_{b\in\B_{2d}} a_b b,\qquad
 \deg a_b\le2d-\deg b,\qquad
 \norm{a_b}_1\le W\norm q_1(1+WC)^W.
\]
At most $W$ cancellation steps and hence at most $W$ basis elements are
used. Regard the resulting identity as a certificate with these basis
elements added as equality axioms. Theorem~\ref{thm:transfer}\textup{(i)}
eliminates them with degree~\eqref{eq:dual-degree} and a coefficient
bound polynomial in the representation size and the magnitudes already
bounded above. Write $W_\Delta=\binom{n+\Delta}{\Delta}$.
Since $\Delta=O(d)$,
$W_\Delta\le W^{O(1)}$: split any monomial of degree at most $\Delta$ into a fixed
number of factors of degree at most $2d$, choosing a deterministic split.
This proves the asserted coefficient bound, with computable majorants.

If $S=\varnothing$, its reduced basis is $\{1\}$. Skip the averaging
lemma and use $q=f=a_1\cdot1$ directly in~\eqref{eq:signed-product}.
The same degree and coefficient conclusions follow from the two assumed
proofs of $1$ and $-1$.

Finally use bounded-certificate search: a computable magnitude bound
on a real certificate suffices to find a rational certificate of
$f+\varepsilon\ge0$ at the same degree, with polynomial cost in the
monomial count, input encodings, bound encoding and accuracy bits.
Lemma~\ref{lem:certificate-search} proves this by finding a rational
point in a small relaxation of the coefficient SDP and repairing its
residual with Boolean equations. Its input consists only of the original
coefficient equations and the computed bound, so the search does not
require the basis or its two-sided proofs.
\end{proof}

The criteria include refutations by taking $f=-1$. For proof search,
choose $\varepsilon=1/2$ and multiply the returned certificate of
$-1/2\ge0$ by two. This statement is subject to the hypotheses of the
dual criterion; for Min-closed systems, Lemma~\ref{lem:propagation}
constructs a degree-four refutation whenever the system is infeasible.

\section{The primal criterion: identities determine the right geometry}
\label{sec:primal}
The primal argument starts from the vector space
\[
\I(S)\cap\Q[x]_{\le 2d}
\]
of all polynomial identities of degree at most $2d$ that vanish on the
feasible set. The geometric criterion assumes that this space is available
through a rational spanning list that can be constructed efficiently.
We then describe two ways to produce such a list: a direct construction,
and graded division by a suitable truncated Gr\"obner basis.
Both constructions impose the same affine equalities on the moments and
therefore define the same augmented body $K_{2d}$.

This is related to quotient-ring moment formulations: combinatorial
moment matrices over finite varieties are described
in~\cite[Section~2.1]{Laurent2007}, and degree-compatible quotient bases
underlie the theta-body construction in~\cite[Section~2]{GPT10}.
Here $K_{2d}$ also retains the localizing constraints for the original
inequalities.

The underlying moment-solvability principle belongs to the algebraic
framework of Gribling--Polak--Slot~\cite{Gribling23}. The theorem below
gives a constructive Boolean specialization, with an explicit affine
reduction, quantitative radius and separation bounds, and an exactly
feasible rational output. When the same identities admit two-sided SoS
derivations, the transfer theorem additionally relates this augmented
moment body to the original hierarchy.

\begin{theorem}[Primal criterion]
\label{thm:main}
Let $S\ne\varnothing$ be the Boolean feasible set of~\eqref{eq:system}.
Consider the following two properties.
\begin{enumerate}[label=\textup{(\roman*)}]
\item\label{hyp:construction}
One can construct, from the system and $d$, a finite list
$Q_{2d}\subseteq\Q[x]_{\le2d}$ satisfying
\begin{equation}\label{eq:spanning}
 \spanQ Q_{2d}=\I(S)\cap\Q[x]_{\le2d},
\end{equation}
in time polynomial in $W,\mathsf{len}$, with total rational encoding length
$\mathsf{len}_Q=\operatorname{poly}(W,\mathsf{len})$.
\item\label{hyp:derivation}
For every $q\in Q_{2d}$, both $q\ge0$ and $-q\ge0$ have SoS proofs
from the original system of degree at most $\kappa$, where
$\kappa\ge2d$ is even and $\kappa=O(d)$.
\end{enumerate}
Under hypothesis~\ref{hyp:construction}, define
\begin{equation}\label{eq:K}
 K_{2d}=\{y\in S_{2d}:L_y(q)=0\text{ for every }q\in Q_{2d}\}.
\end{equation}
For every $f\in\Q[x]_{\le2d}$ and every rational $\varepsilon>0$, one
can compute a rational $y^*\in K_{2d}$ satisfying
\begin{equation}\label{eq:optimization}
 L_{y^*}(f)\ge\max_{y\in K_{2d}}L_y(f)-\varepsilon.
\end{equation}
The running time and output length are polynomial in
$W,\mathsf{len},\bits f,\bits\varepsilon$. For fixed $d$ and $\varepsilon=2^{-s}$,
this is polynomial in the input length and $s+1$. All normalization,
equality and PSD constraints are satisfied exactly.

If hypothesis~\ref{hyp:derivation} also holds, then
\begin{equation}\label{eq:sandwich}
 \pi_{2d}(S_\kappa)\subseteq K_{2d}\subseteq S_{2d},
 \qquad
 L_{y^*}(f)\ge\max_{u\in S_\kappa}L_u(f)-\varepsilon.
\end{equation}
Consequently $\widetilde{\mathbb E}[p]=L_{y^*}(p)$, for $\deg p\le2d$,
is the required rational pseudoexpectation.
\end{theorem}

The list in~\eqref{eq:spanning} spans all ordinary vanishing polynomials
through degree $2d$. Multilinearization expresses their equalities in
moment coordinates. Construction supplies the optimization algorithm;
SoS soundness supplies the higher-level comparison, as shown in
Figure~\ref{fig:primal-roles}.

\begin{figure}[tbp]
\centering
\begin{tikzpicture}
\node[paperbox,text width=6.0cm,minimum height=1.15cm] (construct) at (-3.65,0)
  {\textbf{Constructibility}\\
   Compute $Q_{2d}$ spanning $\I(S)_{\le2d}$};
\node[paperbox,text width=6.0cm,minimum height=1.15cm] (derive) at (3.65,0)
  {\textbf{Derivability}\\
   SoS proofs of $\pm q\ge0$ in degree $\kappa$};
\node[paperbox,text width=6.0cm,minimum height=1.35cm,below=8mm of construct]
  (geometry)
  {$y=y^0+Vz$, inner / outer balls\\
   Strong separation and rational ellipsoids};
\node[paperbox,text width=6.0cm,minimum height=1.35cm,below=8mm of derive]
  (soundness)
  {For $u\in S_\kappa$: $L_u(q)=0$\\
   for every $q\in Q_{2d}$};
\node[paperbox,text width=6.0cm,minimum height=1.35cm,below=8mm of geometry]
  (output)
  {Compute $y^*\in K_{2d}\cap\Q^{\J_{2d}}$\\
   $L_{y^*}(f)\ge\max_{K_{2d}}L_y(f)-\varepsilon$};
\node[paperbox,text width=6.0cm,minimum height=1.35cm,below=8mm of soundness]
  (inclusion)
  {Restriction preserves the identities\\
   $\pi_{2d}(S_\kappa)\subseteq K_{2d}\subseteq S_{2d}$};
\node[paperbox,text width=12.6cm,minimum height=1.45cm,anchor=north]
  (combined) at ($(output.south)!0.5!(inclusion.south)+(0,-.85)$)
  {$L_{y^*}(f)\ge\max_{u\in S_\kappa}L_u(f)-\varepsilon$\\[3pt]
   Min-closed: $\kappa=8d+4$,\quad
   $\pi_{2d}(S_{8d+4})\subseteq K_{2d}\subseteq S_{2d}$};
\draw[paperarrow] (construct) -- (geometry);
\draw[paperarrow] (geometry) -- (output);
\draw[paperarrow] (derive) -- (soundness);
\draw[paperarrow] (soundness) -- (inclusion);
\draw[paperarrow] (output.south) -- (combined.north -| output.south);
\draw[paperarrow] (inclusion.south) -- (combined.north -| inclusion.south);
\end{tikzpicture}
\caption{The two hypotheses of the primal criterion, for the same complete
list $Q_{2d}$ and nonempty $S$. Constructibility gives exact-feasible
optimization over $K_{2d}$; derivability gives the displayed containment.
Their combination gives the objective comparison. For Min-closed systems
the comparison degree is $8d+4$, while the algorithm stores only
degree-$2d$ moments. Here $\I(S)_{\le2d}$ abbreviates its rational
degree-$2d$ part.}
\label{fig:primal-roles}
\end{figure}
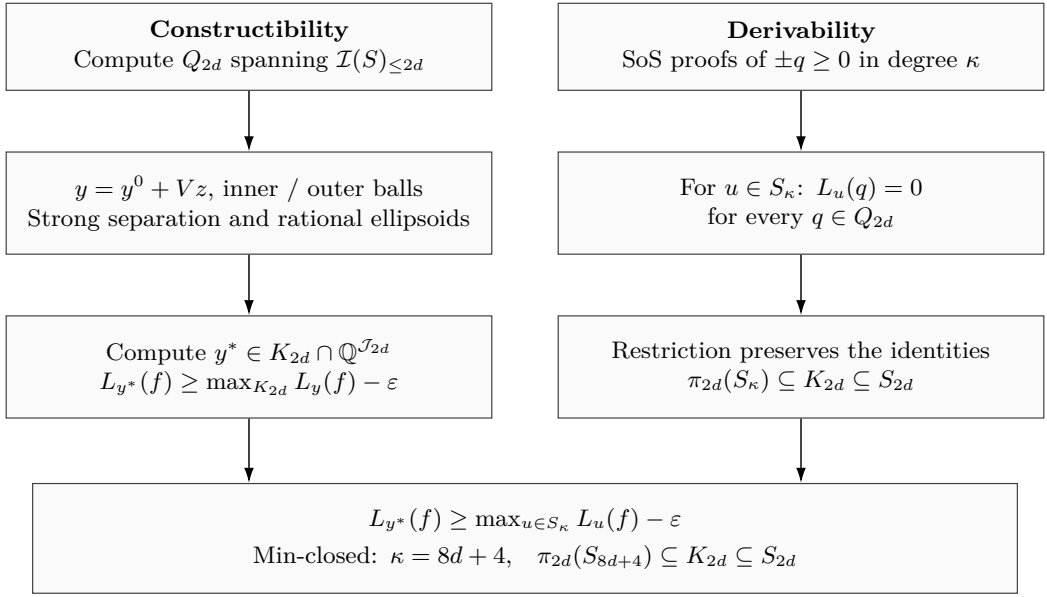

\subsection{Proof of the primal criterion}
\label{sec:proof}

\begin{proof}[Proof of Theorem~\ref{thm:main}]
The proof follows a geometric sequence: the complete identities give
independent moment coordinates; Boolean evaluations supply a simplex
with a controlled inner radius; moment bounds and exact PSD tests give
an outer radius and strong separation. The ellipsoid method then
optimizes in these coordinates without needing the simplex or its
centre. The final step uses two-sided derivability only to compare
hierarchy levels. We give each construction and its bit bound below.

\paragraph{1. Pass to moment coordinates and identify the affine hull.}
Construct $Q_{2d}$ and multilinearize its elements. Put
\[
 \mathcal V_{2d}(S)=
 \{h\in\I(S)\cap\Q[x]_{\le2d}:h\text{ is multilinear}\}.
\]
Equation~\eqref{eq:spanning} gives
\begin{equation}\label{eq:mlspan}
 \spanQ\{\ml(q):q\in Q_{2d}\}=\mathcal V_{2d}(S).
\end{equation}
Indeed, each reduced polynomial is an identity. Conversely, if
$h\in\mathcal V_{2d}(S)$, express $h$ as a rational linear combination
of $Q_{2d}$ and apply $\ml$; the left side remains $h$.
Grouping coefficients by monomial support constructs these reductions
in polynomial time and bit length in $W,\mathsf{len}_Q$.

Let $A$ have their coefficient vectors as rows, indexed by $\J_{2d}$,
and let $\Haff=\{y:Ay=0,\ y_\varnothing=1\}$. For the proof only,
let $E$ have rows $e(s)^\top=(x_T(s))_{T\in\J_{2d}}^\top$, $s\in S$.
Its rational kernel is the coefficient space of $\mathcal V_{2d}(S)$.
Since $E$ is a rational matrix, its real kernel has a rational basis.
Thus~\eqref{eq:mlspan} implies
\[
 \ker_{\R}E=\operatorname{rowspan}_{\R}A,
 \qquad \ker A=\spanR\{e(s):s\in S\}.
\]
Every $e(s)$ has empty-monomial coordinate one. Imposing
$y_\varnothing=1$ therefore makes the coefficients in a linear
combination of these vectors sum to one. Consequently
\begin{equation}\label{eq:hull}
 \Haff=\aff\{e(s):s\in S\},\qquad
 \conv\{e(s):s\in S\}\subseteq K_{2d}\subseteq\Haff,
 \qquad \aff(K_{2d})=\Haff.
\end{equation}
Evaluation at a feasible point satisfies all PSD constraints and
identities, proving the middle inclusion. Every original equality
multiple of degree at most $2d$ vanishes on $S$, so its multilinear
reduction is also in the row span of $A$.

The matrix $E$ and the feasible set are not enumerated. Different complete lists
$Q_{2d}$ define the same $K_{2d}$. Any equality multiple of ordinary
degree at most $2d$ is already imposed, because it is itself a truncated
vanishing identity.

\paragraph{2. Parametrize using free original moments.}
Let $q_0=\rank A$ and $p=N-q_0-1\ge0$.
Rational elimination constructs rational $y^0,V$ and distinct
nonempty sets $T_1,\ldots,T_p\in\J_{2d}$ such that
\begin{equation}\label{eq:Phi}
 \Phi(z)=y^0+Vz:\R^p\longrightarrow\Haff
 \text{ is a bijection},\qquad \Phi(z)_{T_i}=z_i.
\end{equation}
Its rank premise holds because no nonzero constant polynomial vanishes
on the nonempty set $S$: deleting the empty-monomial column of $A$
preserves row rank. Choose $q_0$ nonempty pivot columns $P$ and
$q_0$ rows $I$ with $A_{I,P}$ nonsingular. After setting
$y_\varnothing=1$ and $y_{T_i}=z_i$, the formula is
\[
 y_P=-A_{I,P}^{-1}A_{I,\varnothing}
                  -\sum_{i=1}^p A_{I,P}^{-1}A_{I,T_i}z_i.
\]
When $q_0=0$, there are no pivot equations.

The dense matrix encoding $\mathsf{len}_A$ is polynomial in $W,\mathsf{len}_Q$. Clearing
denominators gives integer entries of $O(\mathsf{len}_A+1)$ bits, and all minors
have $O(N(\mathsf{len}_A+\log(N+1)))$ bits. The exact elimination and cofactor
argument in Appendix~\ref{app:linear} give polynomial construction time and
polynomial bit lengths for $y^0,V$. In particular,
\begin{equation}\label{eq:actualvertices}
 v(s)=(x_{T_1}(s),\ldots,x_{T_p}(s))\in\{0,1\}^p,
 \qquad \Phi(v(s))=e(s).
\end{equation}
If $p=0$, $\Haff=\{y^0\}$ contains every actual evaluation and is
therefore feasible. Return $y^0$, which is exactly optimal. Assume
$p\ge1$ in the next four steps.

\paragraph{3. Form the reduced SDP.}
Let $v_0=y^0$ and let $v_i$ be column $i$ of $V$, $1\le i\le p$.
Substitution in the PSD matrices gives
\begin{align*}
 (F_{0,i})_{A,B}&=(v_i)_{A\cup B},&&A,B\in\J_d,\\
 (F_{j,i})_{A,B}&=\sum_U g_{j,U}(v_i)_{U\cup A\cup B},
 &&A,B\in\J_{\lfloor(2d-\deg g_j)/2\rfloor},
\end{align*}
where $\ml(g_j)=\sum_U g_{j,U}x_U$ and only $\deg g_j\le2d$ is
retained. The block size uses the original ordinary degree of $g_j$.
Define
\begin{equation}\label{eq:F}
 F_j(z)=F_{j,0}+\sum_{i=1}^p z_iF_{j,i},\qquad
 \widehat K_{2d}=\{z:F_j(z)\succeq0\text{ for every retained }j\}.
\end{equation}
All equalities hold identically under $\Phi$, so
$\Phi:\widehat K_{2d}\to K_{2d}$ is an affine bijection.
There are at most $(m+1)(p+1)N^2$ scalar table entries. Each is a sum
of at most $W$ products of an input coefficient and an entry of $y^0,V$.
If $H_V$ bounds the latter bit lengths, an entry has at most
$O(W(\mathsf{len}+H_V+\log(W+1)))$ bits. Thus the rational coefficient table,
including dimensions, has length $\mathsf{len}_K=\operatorname{poly}(W,\mathsf{len},\mathsf{len}_Q)$
and is computable in that time.

\paragraph{4. Supply inner and outer radii.}
Equations~\eqref{eq:hull} and~\eqref{eq:actualvertices} imply
$\aff\{v(s):s\in S\}=\R^p$. Choose existentially $p+1$ affinely
independent vectors $v^0,\ldots,v^p$ in this set. Their simplex is
contained in $\widehat K_{2d}$. Since they are $0/1$ vectors,
$U=[v^1-v^0\ \cdots\ v^p-v^0]$ is an invertible integer matrix with
entries in $\{-1,0,1\}$ and $|\det U|\ge1$.
Applying Lemma~\ref{lem:ball} to this integral simplex gives
\begin{equation}\label{eq:ball}
 B(z^\circ,\rho_*)\subseteq\widehat K_{2d},\qquad
 z^\circ=\frac1{p+1}\sum_{j=0}^p v^j,\qquad
 \rho_* =\frac1{2p^2(p+1)p!}.
\end{equation}
The estimate $|(U^{-1})_{ij}|\le(p-1)!$ bounds the
change in barycentric coordinates under a perturbation of norm
$\rho_*$. Hence $\log(1/\rho_*)=O(p\log(p+1))$.
Only the numerical radius, computable from $p$, is required;
neither the vertices nor the centre must be found.

Lemma~\ref{lem:outer} gives $|y_T|\le1$ for $y\in S_{2d}$,
using principal $2\times2$ moment matrices and a decomposition of $T$
into two sets of size at most $d$. Because free coordinates are original
moments, $\widehat K_{2d}\subseteq[-1,1]^p\subseteq B(0,p)$.
This set is nonempty, closed and convex by~\eqref{eq:F}, and bounded,
so it is compact. Persistent kernels of the PSD matrices cause no
problem: the ball is in the space of free moments.

\paragraph{5. Supply strong rational separation.}
Given $\bar z\in\Q^p$, evaluate $F_j(\bar z)$ exactly.
The standard exact PSD test, Lemma~\ref{lem:psd}, tests rational PSD matrices in polynomial
bit complexity and, on rejection, produces rational $w$ such that
$w^\top F_j(\bar z)w<0$. If every block is PSD, accept the query.
Otherwise set
\[
 c_0=w^\top F_{j,0}w,\qquad c_i=w^\top F_{j,i}w\quad(i\in[p]).
\]
For every $z\in\widehat K_{2d}$,
\[
 (-c)^\top z\le c_0<(-c)^\top\bar z.
\]
The normal is nonzero, since a negative constant quadratic form would
contradict nonemptiness. Each separator entry is a sum of at most $N^2$ products of two witness
entries and one table entry. Matrix evaluation, the PSD test and these sums
therefore have polynomial time and output length in $\mathsf{len}_K+\bits{\bar z}$.
The oracle reads only the rational matrix table; all sign tests are exact.

\paragraph{6. Optimize and reconstruct.}
Write $\ml(f)=\sum_T f_Tx_T$ and $\mathbf f=(f_T)_T$. Then
\[
 L_{\Phi(z)}(f)=\mathbf f^\top y^0+\widehat c^\top z,
 \qquad\widehat c=V^\top\mathbf f.
\]
Compute the powers of two
\[
 R=2^{\lceil\log p\rceil},\qquad
 \rho=2^{\lfloor\log\rho_*\rfloor},\qquad
 \tau=2^{\lfloor\log\min\{1,\varepsilon\}\rfloor}
\]
by exact integer comparisons. They satisfy $R\ge1$, $0<\rho,\tau\le1$,
$\tau\le\varepsilon$, and
$B(z^\circ,\rho)\subseteq\widehat K_{2d}\subseteq B(0,R)$.
The exact-feasible approximate optimization lemma
(Lemma~\ref{ell:lem:main}, proved in Appendix~\ref{app:ellipsoid}) returns
$z^*\in\widehat K_{2d}\cap\Q^p$ with
\[
 \widehat c^\top z^*\ge
       \max_{z\in\widehat K_{2d}}\widehat c^\top z-\tau.
\]
This lemma requires known radius bounds but not a supplied centre.
Set $y^*=y^0+Vz^*$. The bijection gives exact feasibility in $K_{2d}$,
and restoring the objective constant proves~\eqref{eq:optimization}.

For bit complexity, $p\le N-1\le W-1$;
$\mathsf{len}_A,H_V,\mathsf{len}_K$ are polynomial in $W,\mathsf{len},\mathsf{len}_Q$;
$\log R=O(\log(p+1))$ and $\log(1/\rho)=O(p\log(p+1))$;
the objective length is polynomial in $W,\mathsf{len}_Q,\bits f$;
and $\log(1/\tau)\le\bits\varepsilon+1$.
Lemma~\ref{ell:lem:main} and rational reconstruction therefore have polynomial
bit cost and output length in $W,\mathsf{len},\mathsf{len}_Q,\bits f,\bits\varepsilon$.
Hypothesis~\ref{hyp:construction} absorbs construction of $Q_{2d}$ and
$\mathsf{len}_Q$ into $\operatorname{poly}(W,\mathsf{len})$. No step increases the moment degree.

\paragraph{7. Use the SoS derivations to compare hierarchy levels.}
Take $u\in S_\kappa$ and $q\in Q_{2d}$. Soundness and
hypothesis~\ref{hyp:derivation} give
$L_u(q)\ge0$ and $L_u(-q)\ge0$, hence $L_u(q)=0$.
The restriction $y=\pi_{2d}(u)$ belongs to $S_{2d}$ and satisfies
$L_y(q)=L_u(q)=0$, because $\deg q\le2d$. Therefore $y\in K_{2d}$.
This proves the first inclusion in~\eqref{eq:sandwich}; the second
follows from the definition of $K_{2d}$.

The sets $S_{2d}$ and $S_\kappa$ are nonempty, closed and bounded by
the Boolean moment bound, so their maxima exist. Projection preserves
the objective $f$, and thus
\begin{equation}\label{eq:chain}
 \max_{s\in S}f(s)\le\max_{u\in S_\kappa}L_u(f)
 \le\max_{y\in K_{2d}}L_y(f)\le\max_{y\in S_{2d}}L_y(f).
\end{equation}
Combining~\eqref{eq:chain} with~\eqref{eq:optimization} proves the
remaining assertion.
\end{proof}

\paragraph{Relation to the GPS criterion.}
For fixed $d$, Proposition~6 of~\cite{Gribling23} applies to the
degree-$2d$ formulation augmented by $Q_{2d}$. Here only the original
axioms of degree at most $2d$ are retained, as in the definition of
$S_{2d}$; the full solution set $S$ supplies the following witness.
To check the hypotheses, take
$L(p)=|S|^{-1}\sum_{s\in S}p(s)$. If $g\in\{1,g_1,\ldots,g_m\}$ and
$\deg(gp^2)\le2d$, then nonnegativity on $S$ gives
\[
 L(gp^2)=0\quad\Longrightarrow\quad gp^2\in\I(S)_{\le2d}
       =\spanR Q_{2d}.
\]
The resulting constant-multiplier representation satisfies their
degree restriction for the equality axioms. Each moment of $L$ is an
integer divided by $|S|\le2^n$, so Theorem~7 of~\cite{Gribling23}
supplies the required positive-eigenvalue bounds, also for the
localizing matrices. Their explicit ball constraint can be added
without changing the relaxation, since
\[
 n-\sum_i x_i^2
   =\sum_i(1-x_i)^2-2\sum_i(x_i^2-x_i).
\]
Multiplication by any square of degree at most $2d-2$ preserves the
degree-$2d$ SoS bound, so all its localizing conditions are already
implied. The same spectral argument applies to this additional block.
Thus the known solvability criterion covers the augmented body.
The proof above constructs its affine coordinates and radius bounds
directly from $Q_{2d}$, without computing $L$ or enumerating $S$.
Its ellipsoid reduction, detailed in Appendix~\ref{app:ellipsoid}, is
included as a self-contained alternative proof of this solvability
conclusion, rather than as an additional hypothesis needed to apply GPS.
The comparison with $S_\kappa$ additionally uses the two-sided
derivations, through SoS soundness.

\begin{remark}[The body optimized]\label{rem:scope}
The algorithm optimizes $K_{2d}$ at moment degree $2d$. Two-sided
derivations give the comparison $\pi_{2d}(S_\kappa)\subseteq K_{2d}$.
Equality with $S_{2d}$, or an extension of an output to $S_\kappa$,
would require additional hypotheses. The effective identity description
is supplied by construction or by the graded-basis route below.
\end{remark}

\subsection{Two effective routes to the same affine hull}
\label{sec:primal-routes}
Throughout this subsection the target moment degree is $r=2d$ and
$S\ne\varnothing$. For the fixed polynomial system, define
\[
 K_r=S_r\cap\aff\{e_r(s):s\in S\},
 \qquad e_r(s)=(x_T(s))_{T\in\J_r}.
\]
This definition is independent of the chosen identity representation.

\paragraph{Route I: construct the truncated identity space directly.}
Construct a rational list $Q_r$ with
$\spanQ Q_r=\I(S)\cap\Q[x]_{\le r}$. Multilinearize its rows, impose
normalization, and use rational elimination as in the proof of
Theorem~\ref{thm:main}. The resulting affine space is exactly
$\Haff_r=\aff\{e_r(s):s\in S\}$, and the augmented SDP is $K_r$.
The construction time and encoding length of $Q_r$ determine the
algebraic preprocessing cost.

\paragraph{Route II: use a supplied or computed graded basis.}
Here the algebraic input is a correct degree-$D$ truncation of the reduced
graded basis of the \emph{vanishing ideal} $\I(S)$, with $D\ge r$.
Graded division through degree $r$ extracts the required identity space.
The use of truncated Gr\"obner bases to make the algebraic input of the
Raghavendra--Weitz criterion effective, and to formulate theta-body
relaxations, appears in~\cite{Mastrolilli}. Gribling--Polak--Slot also
identify Gr\"obner bases as a sufficient source of their algebraic
condition~\cite[discussion following Proposition~6]{Gribling23}.
Here we spell out the degree-preserving extraction from the supplied
vanishing-ideal basis and its rational encoding cost. The following
theorem covers both $D=r$ and a larger supplied truncation.

The normal-form coordinates used below are the standard quotient-basis
coordinates of combinatorial moment matrices
\cite[Section~2.1]{Laurent2007}; see also the degree-compatible bases
in~\cite[Section~2]{GPT10}. We make the truncated identity extraction
and its bit complexity explicit.

\begin{theorem}[Supplied-Gr\"obner-basis primal optimization]
\label{thm:supplied-route}
Let $S\ne\varnothing$ and $r=2d\ge2$. Suppose the correct reduced graded lexicographic basis
truncation $\B_D$ of $\I(S)$ is supplied, with $D\ge r$ and total rational
encoding length $\mathsf{len}_B$. Set
\[
 \B_r=\{b\in\B_D:\deg b\le r\},\qquad
 Q_r^{\rm mult}=\{x^\alpha b:b\in\B_r,\ |\alpha|+\deg b\le r\}.
\]
Then the following statements hold.
\begin{enumerate}[label=\textup{(\roman*)}]
\item The list $Q_r^{\rm mult}$ spans $\I(S)\cap\Q[x]_{\le r}$.
Adding $\B_r$ as equality axioms, with all degree-allowed multipliers,
has moment formulation exactly $K_r$.
\item Graded normal forms supply an alternative complete identity list
and an explicit rational bijection $y=y^0+Vz$ onto $\Haff_r$. Each free
coordinate is an original moment indexed by a standard monomial.
\item For $f\in\Q[x]_{\le r}$ and rational $\varepsilon>0$, one
can compute a rational $y^*\in K_r$ with
$L_{y^*}(f)\ge\max_{y\in K_r}L_y(f)-\varepsilon$, in time and output
length polynomial in $W,\mathsf{len},\mathsf{len}_B,\bits f,\bits\varepsilon$.
All degree-$r$ feasibility conditions hold exactly.
\end{enumerate}
The supplied basis is promised to be correct. The moment degree remains
$r$ for every $D\ge r$.
\end{theorem}
\begin{proof}
\emph{Identity extraction.} In a graded order, every basis element has
degree equal to its leading-monomial degree. A divisor usable in the
division of a polynomial of degree at most $r$ therefore belongs to
$\B_r$. Division never raises degree. Lemma~\ref{lem:division} now proves
the spanning assertion. There are at most $W|\B_r|$ admissible multiples;
they are constructed just by shifting exponents. Pairing with all
polynomial multipliers of permitted degree is equivalent to pairing with
these monomial multipliers. Hence the augmented certificate system has
exactly the asserted moment equalities. The affine-hull argument in
Theorem~\ref{thm:main} identifies its feasible body with $K_r$.

\emph{Normal forms and coordinates.} Let $\mathcal M_r$ be the monomials
of degree at most $r$ not divisible by any leading monomial in $\B_r$.
They are precisely the standard monomials of the full basis in this
degree range. They are squarefree because $x_i^2$ belongs to the initial
ideal of $\I(S)$, and they include $1$ because $S$ is nonempty.
For $\deg p\le r$, write $\operatorname{NF}(p)$ for its reduced normal
form, computed using $\B_r$. It is a linear combination of $\mathcal M_r$
of degree at most $\deg p$, and
\begin{equation}\label{eq:nf-kernel}
 p\in\I(S)\quad\Longleftrightarrow\quad\operatorname{NF}(p)=0.
\end{equation}
Uniqueness and linearity follow within degree $r$: the difference of two reduced remainders would be an ideal polynomial
supported on standard monomials. A nonzero such polynomial has a standard
leading monomial, contradicting the defining Gr\"obner property. This
also proves linearity and~\eqref{eq:nf-kernel}.

For all ordinary $|\alpha|\le r$, form
\[
 Q_r^{\rm NF}=\{x^\alpha-\operatorname{NF}(x^\alpha):|\alpha|\le r\},
\]
omitting zero polynomials. Each is an identity. If $h=\sum h_\alpha
x^\alpha$ is an identity of degree at most $r$, linearity gives
$h=\sum h_\alpha(x^\alpha-\operatorname{NF}(x^\alpha))$.
Thus $Q_r^{\rm NF}$ is another complete spanning list.

Write $\mathcal M_r=\{1,u_1,\ldots,u_p\}$ and, for $T\in\J_r$, compute
\begin{equation}\label{eq:nf-coordinates}
 \operatorname{NF}(x_T)=a_{T,0}+\sum_{j=1}^p a_{T,j}u_j,
 \qquad y_T=a_{T,0}+\sum_{j=1}^p a_{T,j}z_j.
\end{equation}
If $x_T=u_j$, this says $y_T=z_j$; for $T=\varnothing$ it says
$y_\varnothing=1$. Hence the map is injective. Conversely, assign
arbitrary real values to $z_1,\ldots,z_p$ and value one to the constant
standard monomial. Extend linearly to a functional $\lambda$ on their
span and define $\Lambda(p)=\lambda(\operatorname{NF}(p))$. Boolean reduction
does not change normal forms, so $\Lambda$ is represented by the moments in
\eqref{eq:nf-coordinates}. Equation~\eqref{eq:nf-kernel} shows that it
annihilates every truncated identity. Thus its moment vector belongs
to $\Haff_r$. Every vector in $\Haff_r$ satisfies these same normal-form
identities and is reconstructed from its standard moments. This proves
the affine bijection; in particular there are no remaining affine
dependencies among these $p$ coordinates. PSD feasibility is imposed
after this parametrization, exactly as in~\eqref{eq:F}.

\emph{Bit complexity and optimization.} Reading and filtering $\B_D$
costs polynomial time in $\mathsf{len}_B$. Each of the at most $W$ divisions takes
at most $W$ cancellations, with polynomial rational bit length by
Lemma~\ref{lem:division}. The normal-form table, both identity lists,
and $y^0,V$ therefore have polynomial construction time and encoding
length in $W,\mathsf{len}_B$. The standard Boolean coefficient estimate of
Lemma~\ref{lem:basisbits} also bounds the length of the retained reduced
basis by a polynomial in $W$. In normal-form coordinates the feasible
evaluation vectors have entries $u_j(s)\in\{0,1\}$. The same integral
simplex, outer-radius and separation arguments from
Theorem~\ref{thm:main} apply. Its optimization proof, counting $\mathsf{len}_B$
as part of the input, proves~\textup{(iii)}.
\end{proof}

\begin{corollary}[The same-degree supplied-basis result]\label{cor:supplied}
Supplying $\B_{2d}$ as additional equality axioms to degree-$2d$ SoS,
with all allowed multipliers, makes its moment formulation optimizable
in polynomial time with exact rational feasibility and arbitrary
additive accuracy at the same moment degree.
\end{corollary}
\begin{proof}
Apply Theorem~\ref{thm:supplied-route} with $D=r=2d$.
\end{proof}

The empty case is also decidable under the supplied-basis promise:
$S=\varnothing$ if and only if $1\in\B_D$. The reduced basis of the
unit ideal is $\{1\}$, and its constant element is present at every
truncation degree. Thus the supplied-basis algorithm can first perform
this test and report infeasibility before applying the nonempty case.

\paragraph{Three different degrees.}
The target degree $r=2d$ determines the moment variables. The algebraic
input degree $D\ge r$ only specifies how much of the graded basis is
available: elements above degree $r$ are discarded for this construction.
If $\B_D$ is computed rather than supplied, its computation cost must
also be included; a polynomial-time computation for $D=O(d)$ preserves
the fixed-$d$ polynomial-time conclusion. An unbounded input degree is
accounted for by $\mathsf{len}_B$.
The third degree, $\kappa$, bounds SoS derivations and is used solely for
the inclusion $\pi_r(S_\kappa)\subseteq K_r$ proved below.

For example, a degree-$d$ truncation need not determine all identities
needed by a degree-$2d$ moment body. For the Min-closed system
$x_1+x_2\le1$ on the Boolean cube, there is no nonzero degree-one
vanishing identity, whereas $x_1x_2$ is a degree-two identity.
Thus at moment degree two the degree-one truncated basis misses a
required affine equation. The relevant degree is the degree of all
stored moments, not just the order of the moment matrix.

Figure~\ref{fig:two-routes} summarizes the common geometric endpoint.
Different complete lists and different choices of representatives or
standard moments can change the coordinate description, but not
$\Haff_r$ or $K_r$.

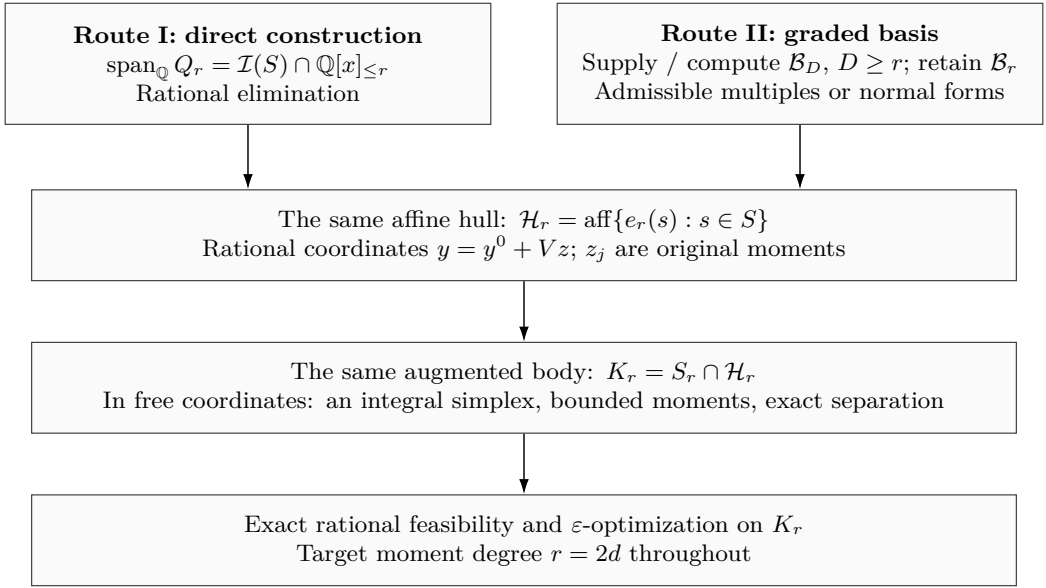
\begin{figure}[tbp]
\centering
\begin{tikzpicture}
\node[paperbox,text width=6.0cm,minimum height=1.6cm] (direct) at (-3.65,0)
  {\textbf{Route I: direct construction}\\
   $\spanQ Q_r=\I(S)\cap\Q[x]_{\le r}$\\
   Rational elimination};
\node[paperbox,text width=6.0cm,minimum height=1.6cm] (basis) at (3.65,0)
  {\textbf{Route II: graded basis}\\
   Supply / compute $\B_D$, $D\ge r$; retain $\B_r$\\
   Admissible multiples or normal forms};
\node[paperbox,text width=12.6cm,minimum height=1.2cm,anchor=north]
  (hull) at ($(direct.south)!0.5!(basis.south)+(0,-.85)$)
  {The same affine hull: $\Haff_r=\aff\{e_r(s):s\in S\}$\\
   Rational coordinates $y=y^0+Vz$; $z_j$ are original moments};
\node[paperbox,text width=12.6cm,minimum height=1.2cm,below=8mm of hull]
  (body)
  {The same augmented body: $K_r=S_r\cap\Haff_r$\\
   In free coordinates: an integral simplex, bounded moments, exact separation};
\node[paperbox,text width=12.6cm,minimum height=1.2cm,below=8mm of body]
  (opt)
  {Exact rational feasibility and $\varepsilon$-optimization on $K_r$\\
   Target moment degree $r=2d$ throughout};
\draw[paperarrow] (direct.south) -- (hull.north -| direct.south);
\draw[paperarrow] (basis.south) -- (hull.north -| basis.south);
\draw[paperarrow] (hull) -- (body);
\draw[paperarrow] (body) -- (opt);
\end{tikzpicture}
\caption{Two implementations of one geometric principle. An appropriate
graded basis supplies the full identity space and can directly supply
standard-monomial coordinates. Both routes optimize $K_r$ at moment
degree $r$. For Min-closed
systems, closures construct either input, and canonical representatives
make the two coordinate descriptions coincide (Proposition~\ref{prop:min-routes}).}
\label{fig:two-routes}
\end{figure}

\subsection{SoS derivability supplies the higher-level comparison}

\begin{corollary}[From two-sided basis derivations to the level comparison]
\label{cor:basiscompare}
Suppose a complete identity list at degree $2d$ is constructible and both
signs of each $b\in\B_{2d}$ have proofs of degree $D_b$.
Set
\[
 \kappa=2\left\lceil\frac12\max\left\{2d,
     \max_{b\in\B_{2d}}[D_b+2(2d-\deg b)]\right\}\right\rceil.
\]
Then every degree-$2d$ vanishing identity has two-sided degree-$\kappa$
proofs, and $\pi_{2d}(S_\kappa)\subseteq K_{2d}$. If $D_b=O(d)$,
this is the comparison in Theorem~\ref{thm:main}.
\end{corollary}
\begin{proof}
Graded division expresses each identity as
$h=\sum_b a_b b$ with $\deg a_b\le2d-\deg b$.
Thus both signs of $h$ have degree-$2d$ certificates when the basis
elements are equality axioms. Apply
Theorem~\ref{thm:transfer}\textup{(i)} to these two certificates and
\textup{(ii)} to the resulting moment constraints. The admissible basis
multiples span the identity space by Theorem~\ref{thm:supplied-route}.
\end{proof}

The two criteria use the same identity space in different ways: bounded
two-sided certificates control proof search for a given target, while an
effective identity description determines the moment geometry. Min-closed
systems supply both ingredients through closures, as we show next.

\section{Application to Min-closed and Max-closed systems}
\label{sec:min}
Let $\mathcal C=\{C_i(x)\ge0:i\in[m]\}$ be a rational linear system,
where $C_i(x)=\beta_i+\sum_j\alpha_{ij}x_j$, together with the Boolean equations. Assume each individual row is
Min-closed on $\{0,1\}^n$, under this input promise. Its solution set $S$
is closed under coordinatewise minimum. We verify the two criteria using
the same closure structure, with all degree and bit bounds explicit.
The closure description is part of the Boolean ideal-membership
framework of~\cite{Mastrolilli}. For the present input of rational
linear rows, we give the construction and the SoS derivations in full.
The constructions provide the equality data for the moment criterion;
the derivations provide the transfer to the original rows, with the
explicit comparison degree $8d+4$.

\subsection{Construct the full list \texorpdfstring{$Q_{2d}$}{Q2d} by closures}

Identify a Boolean vector with its support, and write $\mathbf1_D$
for the vector with support $D$. For $T\subseteq[n]$, define
$S_T=\{s\in S:s_j=1\text{ for all }j\in T\}$.
If $S_T$ is empty, put $\cl(T)=\bot$; otherwise put
\[
 \cl(T)=\bigcap_{s\in S_T}\operatorname{supp}(s).
\]
Min-closure implies $\mathbf1_{\cl(T)}\in S_T$; hence $\cl(T)$ is
the least feasible support containing $T$.

The next lemma gives a propagation algorithm for rational Boolean
Min-closed linear systems and bounds its bit complexity.

\begin{lemma}[Effective closure construction]\label{lem:closure}
Each $\cl(T)$ is computable in time polynomial in $n,m,\mathsf{len}$.
In particular, $\cl(\varnothing)=\bot$ detects $S=\varnothing$.
\end{lemma}
\begin{proof}
Use the following exact propagation algorithm.
Maintain a set $P$ of coordinates forced to one, initially $P=T$.
For every row compute its maximum under this restriction:
\[
 U_i(P)=\beta_i+\sum_{j\in P}\alpha_{ij}
                        +\sum_{j\notin P}\max\{\alpha_{ij},0\}.
\]
If some $U_i(P)<0$, return $\bot$. Otherwise, whenever some free
coordinate $j$ and row $i$ satisfy
$U_i(P)-\max\{\alpha_{ij},0\}<0$, add $j$ to $P$ and restart.
If no such coordinate exists, return $P$.

The latter expression is exactly the maximum of that row when
$x_j=0$ is also imposed. Thus every added coordinate must be one
in every solution extending $T$, and an infeasible row proves
$S_T=\varnothing$. At termination, fix a row. Its feasible assignments
extending $P$ form a nonempty Min-closed family. For each free $j$,
the failed forcing test supplies an assignment in that family with
$x_j=0$. Their intersection is $\mathbf1_P$, which consequently
satisfies the row. If no coordinate is free, nonemptiness directly
implies this assertion. It holds for every row, so $\mathbf1_P\in S_T$;
the forcing invariant proves $P=\cl(T)$.

There are at most $n$ additions and $O(mn)$ rational operations per
scan. All tested quantities are sums of input coefficients, with
polynomial bit length in $\mathsf{len}+n$. This proves the bit bound, including
constant rows and the case $m=0$.
\end{proof}

Assume now that $S\ne\varnothing$. Compute the closures of all
$T\in\J_{2d}$. For each distinct finite closure $C$, select a
representative $T_C\in\J_{2d}$ with $\cl(T_C)=C$, choosing
$T_{C_0}=\varnothing$ for $C_0=\cl(\varnothing)$. Define
\begin{align}
 Q^{\mathrm{cl}}_{2d}
 &=\{x_T:\ T\in\J_{2d},\ \cl(T)=\bot\}\nonumber\\
 &\quad\cup\{x_T-x_{T_C}:\ T\in\J_{2d},\ \cl(T)=C\ne\bot\},
 \label{eq:closurelist}\\
 Q^{\mathrm{bool}}_{2d}
 &=\{x^\alpha-x_{\operatorname{supp}(\alpha)}:
                             \alpha\in\N^n,\ |\alpha|\le2d\},
 \label{eq:booleanlist}\\
 Q_{2d}&=Q^{\mathrm{cl}}_{2d}\cup Q^{\mathrm{bool}}_{2d}.
 \label{eq:minlist}
\end{align}
Zero polynomials may be omitted. Representatives have degree at most
$2d$, even when their closures have larger cardinality. Each nonzero
listed polynomial has at most two terms with coefficients in $\{1,-1\}$.

\begin{lemma}[Completeness and bit complexity of the constructed list]
\label{lem:minspan}
The list~\eqref{eq:minlist} satisfies
$\spanQ Q_{2d}=\I(S)\cap\Q[x]_{\le2d}$.
Its construction time and total encoding length are polynomial in $W,\mathsf{len}$.
\end{lemma}
\begin{proof}
For a feasible support $D$,
\[
 x_T(\mathbf1_D)=
 \begin{cases}
 0,&\cl(T)=\bot,\\
 \mathbf1[\cl(T)\subseteq D],&\cl(T)\ne\bot.
 \end{cases}
\]
Indeed $T\subseteq D$ is equivalent to $\cl(T)\subseteq D$ when
a feasible extension exists. Thus all closure polynomials vanish on
$S$; the Boolean polynomials vanish on the entire cube.

First consider a multilinear $h\in\I(S)\cap\Q[x]_{\le2d}$.
Order the distinct finite closures $C_0,\ldots,C_p$ by nondecreasing
cardinality. Subtract rational linear combinations of
$Q^{\mathrm{cl}}_{2d}$ to remove impossible monomials and replace each
remaining monomial by its representative. The residual is
$h'=\sum_{i=0}^p a_i x_{T_{C_i}}$ and still vanishes on $S$.
Evaluate it at the feasible supports $C_j$, $0\le j\le p$.
The resulting coefficient matrix is
\[
 Z_{ij}=x_{T_{C_i}}(\mathbf1_{C_j})
        =\mathbf1[C_i\subseteq C_j].
\]
It is upper triangular with diagonal entries one: if $i>j$, either
$C_i$ is larger or the sets are distinct of equal size. Hence
$Z^\top a=0$ implies $a=0$. All multilinear identities are therefore
spanned by $Q^{\mathrm{cl}}_{2d}$.

For an arbitrary ordinary polynomial $h=\sum_\alpha h_\alpha x^\alpha$
in $\I(S)\cap\Q[x]_{\le2d}$, write the exact identity
\[
 h-\ml(h)=\sum_\alpha h_\alpha
            (x^\alpha-x_{\operatorname{supp}(\alpha)}).
\]
The right side is spanned by $Q^{\mathrm{bool}}_{2d}$, while $\ml(h)$
is a multilinear identity and is spanned by $Q^{\mathrm{cl}}_{2d}$.
This proves the full ordinary-polynomial spanning equality.

There are $N$ closure calls, at most $N$ nonzero closure polynomials,
and at most $W$ Boolean polynomials. Closures are stored as $n$-bit
sets and grouped by exact comparison. Every output polynomial has
degree at most $2d$ and at most two unit coefficients, with supports
and exponents of polynomial encoding length in $n,d$.
Lemma~\ref{lem:closure} and $n,d,N\le O(W)$ give polynomial total
construction time and length in $W,\mathsf{len}$. This establishes
hypothesis~\ref{hyp:construction} of the criterion.
\end{proof}

\subsection{Comparing the direct and Gr\"obner-basis routes}
\label{sec:min-routes}
The closure construction yields both a complete identity list and a
truncated reduced graded basis. These give two implementations of the
primal algorithm on the same input system.

\begin{proposition}[Two closure-based implementations for Min-closed systems]
\label{prop:min-routes}
For nonempty rowwise Min-closed systems and $r=2d$, one can construct
either the direct list~\eqref{eq:minlist} or the reduced graded basis
$\B_r$ in time polynomial in $W,\mathsf{len}$. Both give exactly the same
$\Haff_r$ and $K_r$. Choosing the least graded-lexicographic
representative in each finite closure class makes the direct class
coordinates identical to the standard-monomial coordinates of
Theorem~\ref{thm:supplied-route}.
\end{proposition}
\begin{proof}
The direct list and its completeness were proved in
Lemma~\ref{lem:minspan}. For the basis route, compute the same closures
for every squarefree monomial of degree at most $r$. For each finite
closure $C$ among them, choose the least monomial $u_C$ in the fixed
graded lexicographic order. This is also the least monomial in its
\emph{entire} evaluation class. Indeed the least representative is
squarefree, by Boolean reduction, and it cannot have degree greater
than a representative already present through degree $r$.
The independence argument of Lemma~\ref{lem:minspan}, or equivalently
the proof of the standard basis description in
Lemma~\ref{lem:min-basis}, shows that the $u_C$ are exactly the standard
monomials through degree $r$.

Enumerate all ordinary monomials $m=x^\alpha$ of degree at most $r$,
and define
\[
 \nu(m)=
 \begin{cases}
 0,&\cl(\operatorname{supp}\alpha)=\bot,\\
 u_C,&\cl(\operatorname{supp}\alpha)=C\ne\bot.
 \end{cases}
\]
For each nonstandard $m$, meaning $m\ne\nu(m)$, retain $m-\nu(m)$
if all its immediate proper divisors $m/x_i$, for $x_i\mid m$, are
standard. Nonstandard monomials are exactly the monomials in the initial
ideal. Testing these immediate divisors therefore selects precisely its
minimal generators through degree $r$: if any proper divisor is in the
initial ideal, an immediate divisor is also in it, by closure under
multiplication. The tail $\nu(m)$ is zero or a single standard monomial,
and in the latter case is smaller than $m$. The proof of
Lemma~\ref{lem:min-basis} identifies these polynomials as exactly the
corresponding reduced basis elements. Thus the output is $\B_r$,
with no missing or extra elements. Only stored closures and monomial
comparisons are used. There are $W$ monomials and at most $n$ immediate
divisors per monomial, and output coefficients belong to $\{0,1,-1\}$.
This gives the claimed polynomial bit complexity.

For the coordinate comparison, let $C_0=\cl(\varnothing)$.
The direct equations set $y_T=0$ in impossible classes, set $y_T=1$
in class $C_0$, and set $y_T=z_C$ in every other finite class $C$.
Independence of these class functions proves that the parameters $z_C$
are free affine coordinates. Choosing $x_{T_C}=u_C$ gives
$z_C=y_{u_C}$, where this notation denotes the moment indexed by the
squarefree monomial $u_C$. Normal form sends $x_T$ to zero or $u_C$,
so~\eqref{eq:nf-coordinates} gives the very same equations and variables.
The two identity lists consequently have the same row space after
multilinearization, and the same affine hull and augmented body.
\end{proof}

If a basis $\B_D$ with $D>2d$ has already been computed, it can instead
be filtered as in Theorem~\ref{thm:supplied-route}; this yields the same
construction at the target degree. In the Min-closed application
$D=2d$ already suffices. The degree $8d+4$ appearing later is the SoS
derivation degree for the higher-level comparison; either primal
implementation uses identity or basis degree $2d$. 

\subsection{Propagation and infeasibility have degree-four certificates}
\label{sec:propagation}
The closure algorithm also has a low-degree proof interpretation.
The certificates below multiply the original inequalities by Boolean
literals, represented as squares modulo the Boolean equations.
We prove the degree-four bound directly, including its preservation
along an entire propagation chain.

\begin{lemma}[Certificates for propagation and refutation]
\label{lem:propagation}
Run the closure algorithm from the empty support on a rowwise Min-closed
linear system. Every propagated assertion $x_i=1$ has a rational Gram
SoS proof of $x_i-1\ge0$ of degree at most four. If the system is
infeasible, it has a degree-four rational Gram refutation, that is, a
proof of $-1\ge0$. All coefficient bit lengths and construction time
are polynomial in $n,m,\mathsf{len}$.
\end{lemma}
\begin{proof}
First consider a feasible row that forces $x_i=1$ without previous
assignments. Write it as
\[
 C(x)=t-\sum_j w_jx_j\ge0,
 \qquad
 N_-=\{j:w_j<0\},\quad P_+=\{j:w_j>0\},
 \qquad
 s_-=\sum_{j\in N_-}w_j.
\]
The minimum of $\sum_j w_jx_j$ over the Boolean cube is $s_-$.
Since the row is feasible and forces $x_i=1$, we must have $w_i<0$:
otherwise, changing $x_i$ from one to zero would preserve feasibility.
Under the restriction $x_i=0$, the minimum becomes $s_-+|w_i|$.
This restricted row is infeasible, so
\[
 \sigma:=s_-+|w_i|-t>0.
\]
The following congruence modulo the Boolean ideal supplies a proof:
\begin{align}
 \sigma(x_i-1)\equiv{}&(1-x_i)^2 C(x)
       +\sum_{j\in P_+}w_j[x_j(1-x_i)]^2\nonumber\\
       &+\sum_{j\in N_-\setminus\{i\}}|w_j|
                               [(1-x_j)(1-x_i)]^2.
 \label{eq:forcing-certificate}
\end{align}
Indeed, replace every squared Boolean literal by that literal. Terms
containing $x_j$ then cancel, and the remaining coefficient of $1-x_i$
is $t-s_--|w_i|=-\sigma$. Every weight is nonnegative, the localizing
term has degree three, and the other squares have degree at most four.
Lemma~\ref{lem:boolean} lifts this to an exact degree-four identity.

Now let $P$ be the previously forced set. A row maximum under $x_P=1$
is obtained by also setting all its negative-weight variables to one and
its other positive-weight variables to zero. Thus already fixed
negative-weight variables do not affect either its feasibility or its
forcing gap. Only $J=P\cap P_+$ matters for that gap. Define
\[
 C'(x)=C(x)-\sum_{j\in J}w_j(1-x_j).
\]
This deletes the variables in $J$ and subtracts their weights from the
capacity. If the row forces a free $x_i=1$, use
\eqref{eq:forcing-certificate} with $C'$ in place of $C$ and the positive
set $P_+\setminus J$. Replace its localizing term by
$(1-x_i)^2C-\sum_{j\in J}w_j(1-x_j)(1-x_i)^2$. For each $j\in J$,
\begin{equation}\label{eq:propagation-reuse}
 -(1-x_j)(1-x_i)^2\equiv(x_j-1)+[x_i(1-x_j)]^2
       \pmod{\I(\{0,1\}^n)}.
\end{equation}
This follows by multilinearizing both sides and lifts in degree four.
Substitute the already constructed proof of $x_j-1\ge0$ into the first
term on the right. It is multiplied only by the positive scalar $w_j$,
not by a further polynomial. Hence the whole propagation chain remains
of degree four. This proves the assertion by induction on its length.

If a residual row is infeasible, its gap is $\sigma=s_--t'>0$, where
$t'=t-\sum_{j\in J}w_j$. The identity
\begin{equation}\label{eq:row-refutation}
 -\sigma\equiv C'(x)
       +\sum_{j\in N_-}|w_j|(1-x_j)^2
       +\sum_{j\in P_+\setminus J}w_jx_j^2
\end{equation}
is verified by cancellation. Substitute
$C'=C+\sum_{j\in J}w_j(x_j-1)$ and the degree-four propagation proofs,
and divide by $\sigma$. This yields a degree-four proof of $-1$.
Lemma~\ref{lem:closure} ensures that an infeasible system does reach such
a row: otherwise the terminating assignment would satisfy every row.

For the bit bound, first scale each input row by a positive product of
its denominators. All weights and capacities become integers of magnitude
at most $H=2^{O(\mathsf{len}+1)}$. Each positive gap is an integer between $1$ and
$(n+2)H$. There are at most $n$ forcing steps and one refutation step.
If $\Delta_r$ is the product of the first $r$ gaps, all coefficients after
step $r$ have denominators dividing $\Delta_r$: the formulas only add
integer multiples of previous certificates and divide by the current gap.
If $M_r$ bounds their magnitudes, the explicit formulas, sums and Boolean
reductions give, for an absolute constant $c_0$,
\[
 M_r+1\le c_0(n+1)^4H(M_{r-1}+1).
\]
Consequently $\log\Delta_r$ and $\log(M_r+1)$ are polynomial in
$n,\mathsf{len}$. Restoring the original rows multiplies their Gram coefficients
by the positive integer scaling factors, preserving the bound. Only
degree-four polynomials and degree-at-most-two Gram monomials are used,
so the number of stored coefficients and all rational operations are
polynomial in $n,m,\mathsf{len}$. This also covers constant infeasible rows.
\end{proof}

\subsection{The reduced basis and its two-sided certificates}
The standard structural description proved in
Appendix~\ref{app:min-basis} (Lemma~\ref{lem:min-basis}) says that the reduced
graded basis consists of squarefree monomials, differences of two squarefree
monomials, and Boolean equations. We now derive both signs of these
identities from the original inequalities, with an explicit degree bound.
For finite-domain equation systems preserved by a semilattice or a dual
discriminator, Bortolotti--Mastrolilli--Vargas
\cite[Theorems~2.1 and~2.3, and Section~8]{BMV}
use low-degree polynomial-calculus derivations of truncated basis
elements to obtain degree automatability. Here the input consists of
linear inequalities and Boolean equations, and we construct the
two-sided SoS certificates required by our criterion.

\begin{lemma}[Conditioning an infeasible partial assignment]
\label{lem:indicator}
Suppose a partial Boolean assignment on disjoint sets $U,V$ has no
extension in $S$. Let $h=x_U\prod_{j\in V}(1-x_j)$ and $s=|U|+|V|$.
Then $h\ge0$ and $-h\ge0$ have rational Gram proofs of degree at most
$2s+4$ from $\mathcal C$ and the Boolean equations.
\end{lemma}
\begin{proof}
Add the Min-closed assignment rows $x_i-1\ge0$ for $i\in U$ and
$-x_j\ge0$ for $j\in V$. The augmented system is infeasible, so
Lemma~\ref{lem:propagation} gives a degree-four refutation. Multiply
that identity by the square $h^2$. Every term attached to an added
assignment row becomes zero modulo the Boolean ideal, because $h$
contains the corresponding literal. The left side is $-h^2\equiv-h$.
All resulting terms have degree at most $2s+4$, so
Lemma~\ref{lem:boolean} removes the zero terms and restores an exact
proof from the original axioms at that degree. Finally $h\equiv h^2$
gives a proof of the positive sign of degree at most $2s$.

The construction takes time polynomial in $n^{O(s+1)}$ and the input
length. For the uses below, $|V|\le1$, so $h$ has at most two terms.
In that case multiplication of Gram matrices by $h^2$ uses a matrix
with entries $0,1,-1$ and at most two contributions per column.
Together with the bit bound in Lemma~\ref{lem:propagation} and Boolean
reduction, all resulting coefficient bit lengths are polynomial in
$n,\mathsf{len},s$.
\end{proof}

\begin{theorem}[Two-sided SoS derivations of Min-closed bases]
\label{thm:mininput}
Let $\mathcal C$ be rowwise Min-closed and $b$ an element of its reduced
graded vanishing-ideal basis, with $t=\deg b$. Both $b\ge0$ and $-b\ge0$
have rational Gram SoS proofs of degree at most $4t+4$, with coefficient
bit lengths polynomial in $n,\mathsf{len},t$ and total size polynomial in
$n^{O(t+1)}$ and $\mathsf{len}$.
\end{theorem}
\begin{proof}
If $S=\varnothing$, the basis is $\{1\}$: the positive sign is trivial
and the negative sign follows from Lemma~\ref{lem:propagation}.
Assume $S\ne\varnothing$ and use Lemma~\ref{lem:min-basis}.
The two signs of a Boolean basis element are equality axioms.
If $b=x_T$ vanishes, the partial assignment $x_T=1$ is infeasible.
Lemma~\ref{lem:indicator} gives both signs in degree $2t+4\le4t+4$.

For $b=p-q$, write $p=x_P$, $q=x_Q$, with $|P|,|Q|\le t$. The
exact identity
\[
 p-q=p(1-q)-q(1-p)
\]
reduces the claim to two nonnegative functions that vanish on $S$.
For an ordering $Q=\{j_1,\ldots,j_a\}$,
\[
 p(1-q)=\sum_{r=1}^a p\,x_{\{j_1,\ldots,j_{r-1}\}}(1-x_{j_r}).
\]
After Boolean reduction each summand is zero, or a partial-assignment
indicator with at most $|P|+|Q|\le2t$ literals and at most one negative
literal. Each such indicator vanishes on $S$: their sum is zero there
and every summand is pointwise nonnegative. Lemma~\ref{lem:indicator}
therefore proves both signs of each one in degree at most $4t+4$.
Apply the same argument to $q(1-p)$ and combine the signs to obtain
both signs of $b$. All intervening telescoping and Boolean identities
have degree at most $2t$ and lift within the stated bound.

At most $2t$ nonzero indicators are used. Their certificates have the
polynomial bit bound proved in Lemma~\ref{lem:indicator}, and summing
them increases denominator and numerator lengths by at most a polynomial
factor in $t$. The number of monomials and Gram entries at degree
$4t+4$ is $n^{O(t+1)}$. This proves the asserted total size bound.
\end{proof}

\subsection{Both algorithmic conclusions and the exact degree accounting}
At this point the two inputs have been established separately:
Lemma~\ref{lem:minspan} constructs the full identity space, and
Theorem~\ref{thm:mininput} proves its basis identities from the given
rows. Theorem~\ref{thm:main} then supplies optimization of the augmented
moment body. The degree calculation below is the
transfer step that gives proof search over the original system and
the comparison to its higher moment relaxation.

\begin{corollary}[Min-closed proof search]\label{cor:min-dual}
Under the promise that $f\ge0$ has a degree-$2d$ SoS proof from
$\mathcal C$, one can find a rational Gram proof of $f+\varepsilon\ge0$
of degree at most $8d+4$ in time polynomial in $W,\mathsf{len},\bits f,
\bits\varepsilon$. An exact, possibly real, proof of $f$ with that degree
and coefficient magnitudes $2^{\operatorname{poly}(W,\mathsf{len}+\bits f)}$ exists.
\end{corollary}
\begin{proof}
The basis coefficients belong to $\{0,1,-1\}$, so $C=1$.
Theorem~\ref{thm:mininput} supplies $D_b=4t+4$ for $t=\deg b\le2d$
and computable coefficient bounds
$B=2^{\operatorname{poly}(n,\mathsf{len},d)}$.
For every basis element the degree in Theorem~\ref{thm:dual} is at most
\begin{equation}\label{eq:min-degree}
 D_b+2(2d-t)=(4t+4)+2(2d-t)=4d+2t+4\le8d+4.
\end{equation}
Apply that theorem, padding the search to degree $8d+4$ if necessary.
The empty-set case uses the degree-four refutation and is also covered.
\end{proof}

\begin{corollary}[Min-closed pseudoexpectation optimization]
\label{cor:min-primal}
Given a rowwise Min-closed rational linear system, $d\ge1$, a rational
$f$ of degree at most $2d$, and rational $\varepsilon>0$, an algorithm
with running time and output length polynomial in $W,\mathsf{len},\bits f,
\bits\varepsilon$ reports infeasibility if $S=\varnothing$.
Otherwise it constructs the list~\eqref{eq:minlist}, or equivalently
uses the graded-basis route of Proposition~\ref{prop:min-routes}, and returns a
rational $y^*\in K_{2d}$ satisfying
\begin{align*}
 K_{2d}&=S_{2d}\cap\aff\{e(s):s\in S\},\\
 L_{y^*}(f)&\ge\max_{y\in K_{2d}}L_y(f)-\varepsilon,\\
 \pi_{2d}(S_{8d+4})&\subseteq K_{2d}\subseteq S_{2d}.
\end{align*}
In particular $L_{y^*}(f)\ge\max_{u\in S_{8d+4}}L_u(f)-\varepsilon$.
All degree-$2d$ moment, localizing and equality constraints hold exactly.
\end{corollary}
\begin{proof}
Compute $\cl(\varnothing)$ by Lemma~\ref{lem:closure}. If it is $\bot$,
report infeasibility. Otherwise Lemma~\ref{lem:minspan} constructs a
complete identity list with polynomial encoding length. Theorem~\ref{thm:mininput},
Corollary~\ref{cor:basiscompare} and~\eqref{eq:min-degree} prove two-sided
derivations for all these identities in degree $8d+4$. Apply
Theorem~\ref{thm:main}. Equation~\eqref{eq:hull} identifies the affine hull.
The direct implementation constructs only the identity list. Alternatively,
Proposition~\ref{prop:min-routes} constructs $\B_{2d}$ and
Theorem~\ref{thm:supplied-route} gives the identical optimization guarantee
for the same body. In both implementations the higher-degree certificates
are used only to prove the level comparison.
\end{proof}

Thus a single structural analysis supplies both the computable identity
space and its low-degree derivability. The moment optimization takes place
at degree $2d$; the number $8d+4$ records the soundness comparison and the
degree sufficient for proof search. The algorithm stores moments through
degree $2d$.

\subsection{Max-closed systems via Boolean complementation}
\label{sec:complement}
The standard coordinate substitution $x_i=1-z_i$ transfers the polynomial
system, its certificates and its moment functionals. Thus the
Max-closed conclusions use the same algorithmic and degree estimates.

\begin{proposition}[Degree-preserving Boolean complementation]
\label{prop:complement}
All conclusions of Corollaries~\ref{cor:min-dual} and~\ref{cor:min-primal}
hold for systems whose rows are all Max-closed. The transformation preserves
certificate degree, exact rational feasibility, objective values and the
inclusion between projected hierarchy levels. Its encoding cost is
polynomial in the number of monomials at the relevant degree and in the
input length.
\end{proposition}
\begin{proof}
Define the involution $\Theta p(z)=p(\mathbf1-z)$.
If $z,z'$ satisfy a transformed row, $\mathbf1-z$ and $\mathbf1-z'$
satisfy the original Max-closed row. Their coordinatewise maximum is
$\mathbf1-\min(z,z')$, so the transformed row is Min-closed.
For a linear row $t-\sum_iw_ix_i$, substitution gives
$(t-\sum_iw_i)+\sum_iw_iz_i$, with polynomial coefficient bit length.
The Boolean polynomial transforms exactly as
$(1-z_i)^2-(1-z_i)=z_i^2-z_i$.

Substitution commutes with addition and multiplication, does not increase
degree, and sends a square to a square. Applying it to a SoS identity
therefore gives a certificate of the same degree for the transformed
system and target. In Gram form it is a rational congruence transformation
on each monomial vector. Applying it twice recovers the original identity.

For a degree-$r$ pseudoexpectation $L_z$ of the transformed system define
$L_x(p)=L_z(\Theta p)$ for $\deg p\le r$. Normalization is preserved;
nonnegativity on every square and every allowed localizing square follows
from the preceding degree observation, and equality multiples are
preserved. Consequently this defines an affine bijection between the two
original moment relaxations. In explicit multilinear moments it reads
\[
 y^x_T=\sum_{U\subseteq T}(-1)^{|U|}y^z_U,\qquad |T|\le r.
\]
This map is rational, invertible by the same formula, and commutes with
restriction to lower degree. It sends actual evaluation vectors to actual
evaluation vectors and hence their affine hulls to each other. Thus it
also maps the augmented bodies $K_r$ bijectively. With transformed target
$\Theta f$, the objective is preserved exactly. A complete identity list
may be transported by $\Theta$ if an explicit original-coordinate SDP is
desired.

Each ordinary monomial of degree at most $r$ expands into at most
$\binom{n+r}{r}$ distinct terms, with integer coefficient magnitudes
at most $2^r$. Each coefficient thus has at most $r+1$ bits.
Expanding and collecting terms therefore costs polynomial time
in $\binom{n+r}{r}$ and the input length, since $r\le\binom{n+r}{r}$.
The same bounds apply to the Gram congruences. For multilinear moments
there are at most $\sum_{j\le r}\binom nj$ possible terms. All transports
therefore preserve the required polynomial bit complexity.
\end{proof}

\subsection{Generalized packing and covering}
Min-closed rows define the generalized packing class considered here;
Max-closed rows define its generalized covering counterpart. Boolean
complementation transfers the algorithmic results between these classes.

Classical packing and covering systems satisfy the respective closure
promises and are special cases of Theorem~\ref{thm:intro}. The same
conclusion holds for the mixed-sign rows that individually satisfy the
corresponding promise: all rows are Min-closed for generalized packing,
and all are Max-closed for generalized covering. Closure is taken on the
Boolean cube.

For a maximization application whose rounding analysis uses degree-$2d$
pseudoexpectation inequalities, Corollary~\ref{cor:min-primal} supplies
those inequalities exactly. In addition its objective is at least the
true Boolean optimum minus $\varepsilon$, because actual evaluations
belong to every hierarchy level. To conclude an application-specific
approximation ratio, its rounding theorem and the conversion of the
additive optimization error to the required scale must still be supplied.

\phantomsection
\section*{Declaration on the use of AI tools}
\addcontentsline{toc}{section}{Declaration on the use of AI tools}

The paper's main ideas, results, and proofs were developed by the author,
building on the cited literature and the contributions acknowledged below.
AI tools assisted with English editing, the organization and drafting
of proofs, and checks of statements, arguments, and references.
This assistance included all appendices, which collect classical tools
and develop the quantitative specializations used in the paper.
The author takes full responsibility for the final content, including
its correctness, attribution, and originality.

\phantomsection
\section*{Acknowledgments}
\addcontentsline{toc}{section}{Acknowledgments}

Preliminary versions of Lemma~\ref{lem:bounded-modulo}  and Theorem~\ref{thm:dual}
were developed with contributions from
Alex Bortolotti and Luis F.\ Vargas while they were working with the
author on SNSF project 200021-207429. These results are included in this
paper with their permission.

\phantomsection
\addcontentsline{toc}{section}{References}
\bibliographystyle{plain}
\bibliography{sos_criterion}

\appendix

\section{Algebraic tools and coefficient bounds}
\label{app:algebra}
Polynomial division and reduced Gr\"obner bases are classical; see
\cite{CLO}. The Min-closed basis structure
is related to Mastrolilli~\cite{Mastrolilli}, with a closure-based proof
below. We prove the Boolean reductions, spectral estimates and explicit
coefficient bounds used in this paper. These quantitative formulations
and the residual correction specify the degree and bit bounds needed
by our algorithms.

\begin{lemma}[Degree-preserving Boolean reduction]\label{lem:boolean}
For a polynomial $p$ of degree at most $t$, the difference $p-\ml(p)$ is
a sum $\sum_i u_i(x)(x_i^2-x_i)$ with $\deg(u_i(x)(x_i^2-x_i))\le t$.
For rational input the multipliers can be computed in polynomial time in
the number of ordinary degree-$t$ monomials and the input length, with
polynomial coefficient bit length. Also $\norm{\ml(p)}_1\le\norm p_1$.
\end{lemma}
\begin{proof}
For each monomial, lower an exponent $a\ge2$ one unit at a time using
$x_i^a-x_i^{a-1}=x_i^{a-2}(x_i^2-x_i)$, retaining its other factors.
At most $t$ reductions per monomial suffice, none increasing degree.
Add the resulting identities and group equal terms. The coefficients are
sums of at most $t$ times the input number of terms, so rational bit
length and computation time are polynomial. Grouping coefficients by
support also proves the norm inequality.
\end{proof}

The degree-preserving representation below is the division algorithm
for a graded order~\cite[Chapter~2, Section~3]{CLO}, together with the
Gr\"obner-basis membership criterion
\cite[Chapter~2, Section~6, Corollary~2]{CLO}. The proof also records
the elementary magnitude and rational encoding bounds used here.

\begin{lemma}[Degree-preserving graded division with a size bound]
\label{lem:division}
Let $\B_r$ be the degree-$r$ truncation of a monic reduced graded basis
of an ideal $I$. If $q\in I$ and $\deg q\le r$, then
\[
 q=\sum_{b\in\B_r}a_b b,\qquad \deg a_b+\deg b\le r.
\]
Writing $W_r=\binom{n+r}{r}$ and assuming basis coefficients have
magnitude at most $C\ge1$, one may choose
$\norm{a_b}_1\le W_r\norm q_1(1+W_r C)^{W_r}$.
Division uses at most $W_r$ cancellations; for rational data it takes
polynomial bit time in $W_r$ and the supplied basis and polynomial lengths.
\end{lemma}
\begin{proof}
At each step cancel the current leading term by a divisible leading
monomial of a basis element. Because the order is graded, every term
introduced has degree at most the cancelled degree, hence at most $r$.
Every chosen divisor belongs to $\B_r$. The leading monomial strictly
decreases, so at most $W_r$ steps occur. The remainder is zero: a nonzero
remainder in the ideal would have leading monomial divisible by a basis
leading monomial, contrary to reduction. Each cancellation multiplies the
current coefficient $\ell_1$ norm by at most $1+W_r C$. The quotient
coefficients are the cancellation coefficients, grouped into at most
$W_r$ terms, which gives the stated bound. For rational input, a common
denominator for all input coefficients raised to $W_r+1$ bounds the
denominators arising from monic division; the magnitude bound controls
the numerators. Both have polynomial bit length. Real division has the
same magnitude argument without any rationality assertion.
\end{proof}

The following integer-matrix estimate is
\cite[Lemma~6]{RW}.
We include its short proof.

\begin{lemma}[Nonzero spectrum of an integer PSD matrix]
\label{lem:integer-spectrum}
Let $A$ be an integer PSD matrix of order $t\ge1$, with entries of
magnitude at most $U\ge1$. Every positive eigenvalue is at least
$(tU)^{-t}$.
\end{lemma}
\begin{proof}
If the rank is $s>0$, the product of the $s$ positive eigenvalues is,
up to sign, the coefficient of $\lambda^{t-s}$ in
$\det(\lambda I-A)$. It is a nonzero integer, so has magnitude at least
one. Every eigenvalue is at most $\norm A_2\le tU$. Dividing the product
bound by the other $s-1$ eigenvalues gives a lower bound
$(tU)^{-(s-1)}\ge(tU)^{-t}$. The rank-zero case is vacuous.
\end{proof}

The next lemma is a Boolean residual-correction counterpart of
\cite[Lemma~23]{Gribling23}. There the correction uses certificates for
$R^{|\alpha|}-x^\alpha$ from the ball constraint; here Boolean equations
give certificates for $1\pm x_T$ and preserve the prescribed even degree.

\begin{lemma}[Degree-preserving correction of a Boolean residual]
\label{lem:residual}
Let $t\ge2$ be even and let $e$ be a rational polynomial of degree at
most $t$. If $\varepsilon\ge\norm{\ml(e)}_1$, then $\varepsilon+e$
has a rational Gram SoS certificate from the Boolean equations alone
of degree at most $t$, constructible in polynomial bit time in the
degree-$t$ monomial count and the encodings of $e,\varepsilon$.
\end{lemma}
\begin{proof}
For a squarefree monomial $m=x_T$, split $T=A\mathbin{\dot\cup}B$ with
$|A|,|B|\le\lceil |T|/2\rceil$, and set $a=x_A$, $b=x_B$. Modulo the
Boolean ideal,
\[
 1\pm m\equiv\tfrac12\big((a\pm b)^2+(1-a)^2+(1-b)^2\big).
\]
Both sides have degree at most $2\lceil |T|/2\rceil\le t$; the identity
also holds for $T=\varnothing$. Write $\ml(e)=\sum_T e_Tx_T$ and use
\[
 \varepsilon+\ml(e)=\varepsilon-\sum_T|e_T|
       +\sum_{T:e_T\ne0}|e_T|(1+\operatorname{sgn}(e_T)x_T).
\]
The constant and the weights are nonnegative rational numbers, and so
give rational PSD Gram terms. Lift the congruences and $e-\ml(e)$ by
Lemma~\ref{lem:boolean}. There are polynomially many terms and all
operations have polynomial bit complexity.
\end{proof}

\subsection{A bit bound for a supplied Boolean basis}
Interpolation on a finite set underlies the Buchberger--M\"oller
construction of its vanishing-ideal basis~\cite{MollerBuchberger82}.
The elementary determinant argument below gives the coefficient bit
bound for Boolean evaluations; the statement concerns size rather than
an algorithm for finding an unknown feasible set.

\begin{lemma}[Bit bound from Boolean evaluations]\label{lem:basisbits}
For nonempty $S\subseteq\{0,1\}^n$ and $r\ge1$, put
$N_r=\sum_{j=0}^{\min\{n,r\}}\binom nj$ and $W_r=\binom{n+r}{r}$.
Every coefficient of the reduced graded basis truncation $\B_r$ is
rational with numerator and denominator of $O(N_r\log(N_r+1))$ bits.
There are at most $W_r$ basis elements in this truncation, and their
total encoding length, including monomial supports, is polynomial in
$W_r$. This is a size bound, not a construction algorithm.
\end{lemma}
\begin{proof}
The standard monomials of degree at most $r$ are squarefree, because
$x_i^2$ belongs to the initial ideal. Let their number be $t\le N_r$.
They are linearly independent as functions on $S$: otherwise an ideal
polynomial would have a standard leading monomial. Their evaluation
matrix on $S$ thus has column rank $t$. Select existentially $t$ rows
giving an invertible $0/1$ matrix $E_0$.

For $b\in\B_r$, its monic leading term is a monomial $m$ and its tail
is a linear combination of these standard monomials, because the basis
is reduced and the order is graded. Write $b=m-\sum_u c_u u$.
Evaluating on the chosen rows yields $E_0c=v$, where $v$ is the $0/1$
evaluation vector of $m$. Cramer's rule expresses each $c_u$ as a ratio
of integer determinants of magnitude at most $t!$, with nonzero
denominator. The reduced numerator and denominator have
$O(t\log(t+1))$ bits.

Different reduced basis elements have distinct leading monomials;
there are at most $W_r$ monomials through degree $r$, and at most $W_r$
terms per element. Monomial supports have encoding length polynomial
in $n,r$, and $n+1,r+1\le W_r$. This proves the total size bound.
For $S=\varnothing$ the basis is $\{1\}$.
\end{proof}

\subsection{Reduced bases for Min-closed Boolean sets}
\label{app:min-basis}
For the Boolean Min-closed ideals considered here, the closure argument
below shows that the normal form of a monomial is either zero or a
monomial with coefficient one, and identifies the possible elements
of the reduced Gr\"obner basis.

\begin{lemma}[Reduced graded bases for Min-closed Boolean sets]
\label{lem:min-basis}
For nonempty Min-closed $S\subseteq\{0,1\}^n$, every element of its
reduced graded lexicographic basis is a squarefree monomial, a difference
of two squarefree monomials, or a Boolean polynomial $x_i^2-x_i$.
All nonzero coefficients are $1$ or $-1$.
\end{lemma}
\begin{proof}
We give the closure proof of this structural description for
Min-closed Boolean ideals; see also~\cite{Mastrolilli}. Consider all feasible
supports, ordered by nondecreasing cardinality. The functions
$D\mapsto\mathbf1[C\subseteq D]$, one for each feasible support $C$,
are linearly independent: their evaluation matrix on these same supports
is upper triangular with diagonal entries one. Every monomial is either
zero on $S$, or agrees on $S$ with exactly one such function, namely
the one indexed by the closure of its support.

In every nonzero evaluation class choose the least monomial $u_C$ in the
fixed graded order. It is squarefree, since multilinearization gives the
same function with smaller degree whenever a square occurs. Every
nonchosen monomial $m$ is a leading monomial of an identity: use $m$ if
it is zero on $S$, and $m-u_C$ otherwise. Conversely a chosen $u_C$
cannot be the leading monomial of an identity. Independence of the
class functions would require another monomial of its class with
nonzero coefficient in the identity; such a monomial is larger than
$u_C$, contradicting leadingness. The standard monomials are therefore
exactly the $u_C$.

For a reduced basis element with leading monomial $m$, all other terms
are standard. Independence forces its tail to be zero in the zero class,
and to be $-u_C$ otherwise. If the minimal leading generator $m$ contains
a square, then $x_i^2$ divides $m$ and is itself in the initial ideal
because $x_i^2-x_i$ is an identity. Minimality gives $m=x_i^2$.
Its proper divisor $x_i$ must be standard, so its basis element is
$x_i^2-x_i$. Otherwise $m$ is squarefree and its tail is squarefree.
This proves every case. For $S=\varnothing$ the reduced basis is $\{1\}$.
\end{proof}

\section{Rational linear algebra and Boolean moment geometry}
\label{app:linear}
Rational elimination, determinant estimates and Schur complements are
classical tools. We give direct proofs of the precise consequences used
by the primal algorithm: polynomial bit bounds, a uniform radius for an
integral simplex, Boolean moment bounds and a rational negative witness
for a failed PSD test. These bounds apply to the free-coordinate
formulation of the moment SDP.

The following rational-elimination fact is classical; see Edmonds'
algorithm as stated in~\cite{GLS}.

\begin{lemma}[Bit bounds for rational elimination]\label{lem:elimination}
A rational linear system of total encoding length $H$ can be solved,
rank-tested and parametrized in polynomial bit time and output length.
The free variables can be chosen among its original coordinates.
\end{lemma}
\begin{proof}
Clear denominators by their product. Integer entries then have $O(H+1)$
bits. Every $s\times s$ minor has absolute value at most $s!2^{sO(H+1)}$.
Gaussian elimination with reduced fractions has entries given by ratios
of such minors. There are polynomially many arithmetic operations on
polynomial-length integers. Choosing pivot columns and solving the
nonsingular pivot system gives the stated parametrization; Cramer's rule
gives the same bit bound for its coefficients. The zero-rank case needs
no inversion.
\end{proof}

A barycentre argument of this type is used
in~\cite[Lemma~6.2.6]{GLS} to bound an inscribed radius from facet
complexity. The following specialization instead uses integral $0/1$
vertices and proves the explicit radius needed here by determinant
bounds.

\begin{lemma}[Uniform ball inside an integral simplex]\label{lem:ball}
Let $p\ge1$, let $v^0,\ldots,v^p\in\{0,1\}^p$ be affinely independent,
and put $\Delta=\conv\{v^0,\ldots,v^p\}$ and
$U=[v^1-v^0\ \cdots\ v^p-v^0]$. Define
\begin{equation}\label{eq:rho}
 z^\circ=\frac1{p+1}\sum_{j=0}^p v^j,\qquad
 \rho_* =\frac{1}{2p^2(p+1)p!}.
\end{equation}
Then
\begin{equation}\label{eq:ball-inclusion}
 B(z^\circ,\rho_*)\subseteq\Delta,
 \qquad \log(1/\rho_*)=O(p\log(p+1)).
\end{equation}
The lower bound $\rho_*$ is computable from $p$ alone in polynomial
time and has polynomial rational encoding length.
\end{lemma}
\begin{proof}
The matrix $U$ is integral with entries in $\{-1,0,1\}$ and
$|\det U|\ge1$. Each entry of $U^{-1}$ is a signed cofactor divided by $\det U$.
In the determinant expansion of a $(p-1)\times(p-1)$ matrix with
entries in $\{-1,0,1\}$ there are $(p-1)!$ products of magnitude at
most one. Hence
\begin{equation}\label{eq:inverse-bound}
 |(U^{-1})_{ij}|\le(p-1)!.
\end{equation}
For $p=1$, use the empty determinant $1$ and $0!=1$.

Take any $h\in\R^p$ with $\norm h_2\le\rho_*$, and put
$\delta=U^{-1}h$. Then
\begin{align}
 \norm\delta_1
 &\le p^2(p-1)!\norm h_\infty
 \le p^2(p-1)!\norm h_2\label{eq:delta-bound}\\
 &\le p^2(p-1)!\rho_*
 =\frac{1}{2p(p+1)}\le\frac{1}{2(p+1)}.\nonumber
\end{align}
Define the $p+1$ barycentric coefficients
\[
 \lambda_i=\frac1{p+1}+\delta_i\quad(1\le i\le p),\qquad
 \lambda_0=\frac1{p+1}-\sum_{i=1}^p\delta_i.
\]
Their sum is one. Each is at least $1/(2(p+1))$, by
the bound on $\norm\delta_1$. In particular, they are nonnegative.
Using $U\delta=h$, we get
\[
 \sum_{j=0}^p\lambda_jv^j
 =\frac1{p+1}\sum_{j=0}^p v^j
     +\sum_{i=1}^p\delta_i(v^i-v^0)
 =z^\circ+h.
\]
Thus the entire closed ball lies in $\Delta$.

Finally $p!\le p^p$ implies
\[
 \log(1/\rho_*)
 \le1+2\log p+\log(p+1)+p\log p=O(p\log(p+1)).
\]
Computing $p!$ requires $p$ integer multiplications on integers of
$O(p\log(p+1))$ bits. The formula for $\rho_*$ therefore has the
claimed construction cost and encoding bound.
\end{proof}

The next bound is an elementary PSD consequence of the Boolean moment
matrix formulation; see~\cite{Laurent2003} for that
formulation. We give the short argument for our moment convention.

\begin{lemma}[Bounds on Boolean pseudo-moments]\label{lem:outer}
Every $y\in S_{2d}$ satisfies $|y_T|\le1$ for every $T\in\J_{2d}$.
Consequently, for $p\ge1$,
\begin{equation}\label{eq:outer}
 \widehat K_{2d}\subseteq[-1,1]^p\subseteq B(0,p).
\end{equation}
For nonempty $S$, the set $\widehat K_{2d}$ of~\eqref{eq:F} is nonempty,
compact, and convex.
\end{lemma}
\begin{proof}
For $\varnothing\ne A\in\J_d$, the principal submatrix of $M_d(y)$
indexed by $\varnothing,A$ is
\[
 \begin{pmatrix}1&y_A\\y_A&y_A\end{pmatrix}\succeq0.
\]
Its determinant gives $y_A-y_A^2\ge0$, so $0\le y_A\le1$.
This also holds for $A=\varnothing$ by normalization.
For $T\in\J_{2d}$ choose a disjoint decomposition $T=A\cup B$
with $|A|,|B|\le d$. Positive semidefiniteness gives
\[
 |y_T|^2=|M_d(y)_{A,B}|^2
 \le M_d(y)_{A,A}M_d(y)_{B,B}=y_Ay_B\le1.
\]
The same inequality is valid when the two indices coincide.
Each $z_i$ is a selected original moment $y_{T_i}$, so
$\norm z_2\le\sqrt p\le p$. This proves~\eqref{eq:outer}.
The PSD cone is closed and convex and the maps $F_j$ are affine;
therefore $\widehat K_{2d}$ is closed and convex. It is bounded by
\eqref{eq:outer} and nonempty because it contains every $v(s)$.
It is consequently compact.
\end{proof}

Polynomial-time testing of rational positive semidefiniteness by
Gaussian elimination is described in~\cite[proof of Theorem~9.3.30,
pp.~294--295]{GLS}. 
We include the elimination proof and its witness encoding bound.
\begin{lemma}[Exact PSD test with a rational negative witness]
\label{lem:psd}
Given a rational symmetric $t\times t$ matrix $G$, one can in polynomial
bit complexity either certify $G\succeq0$ or find $w\in\Q^t$ with
$w^\top Gw<0$. The witness has polynomial encoding length.
\end{lemma}
\begin{proof}
Use symmetric elimination. The empty matrix is PSD. If a diagonal
entry $G_{ii}$ is negative, $w=e_i$ is a witness. If $G_{ii}=0$ and
$a=G_{ij}\ne0$, set
\[
 w_j=1,\qquad w_i=-\frac{G_{jj}+1}{2a},\qquad w_u=0\ (u\notin\{i,j\}).
\]
Then $w^\top Gw=G_{jj}+2aw_i=-1$. A zero diagonal whose entire row
is zero can be deleted, restoring a zero coordinate on return.
If none of these cases applies and the matrix is nonempty, a positive
diagonal entry can be permuted to the first position. Write
\[
 G=\begin{pmatrix}a&b^\top\\ b&C\end{pmatrix},\qquad
 a>0,\qquad H=C-bb^\top/a.
\]
The identity
\begin{equation}\label{eq:schur}
 (s,v)^\top G(s,v)
 =a\left(s+\frac{b^\top v}{a}\right)^2+v^\top Hv
\end{equation}
shows that $G$ is PSD exactly when $H$ is PSD. A negative witness $v$
for $H$ lifts to $(-b^\top v/a,v)$ for $G$. Recurse on $H$.
There are at most $t$ stages and polynomially many rational operations.

To bound the bits, suppose each input entry has length at most $H_0$.
Clear denominators by their positive product, producing an integer
matrix $\bar G$ with entries of $H_1=O(t^2H_0+1)$ bits.
This positive scaling preserves the sign of quadratic forms.
Every minor has $O(t(H_1+\log(t+1)))$ bits, using the determinant
bound $s!2^{sH_1}$ for an $s\times s$ minor.
If $P$ denotes the indices of the positive pivots already eliminated,
then every remaining Schur entry has the form
\[
 \bar G_{ij}-\bar G_{iP}\bar G_{PP}^{-1}\bar G_{Pj}
 =\frac{\det\begin{pmatrix}\bar G_{PP}&\bar G_{Pj}\\
                             \bar G_{iP}&\bar G_{ij}\end{pmatrix}}
        {\det\bar G_{PP}}.
\]
The denominator is nonzero; all reduced fractions have polynomial
bit length. Deleted zero rows do not affect the formula for the
remaining indices. The terminal witness has at most two nonzero
coordinates, of polynomial bit length by the explicit cases above.
If $Q$ is the complement of $P$, insert zeros for deleted coordinates
in $w_Q$ and recover all pivot coordinates at once by
\[
 w_P=-\bar G_{PP}^{-1}\bar G_{PQ}w_Q.
\]
The cofactor formula and the minor bound give polynomial bit length
for this vector. Equation~\eqref{eq:schur} applied successively
verifies that it is a negative witness for $\bar G$ and hence for $G$.
Exact elimination with fractions reduced at each step therefore has
polynomial bit cost, including recovery of the witness.
\end{proof}

\section{The ellipsoid method and exact-feasible optimization}
\label{app:ellipsoid}
This appendix records the standard rational central-cut ellipsoid theorem
of GLS~\cite[Theorem~3.2.1]{GLS}, and proves the exact-feasible optimization
consequence used in the paper. The rational ellipsoid implementation itself
is the classical external algorithmic theorem; its statement, accepted-point
convention and precise reference are given below. All subsequent reductions
and radius-dependent complexity estimates are proved here. The argument is
also valid when the centre of the inner ball is not supplied.
In this appendix $p$ is the ambient dimension, $\mathsf{len}$ the oracle-instance
length, and $\tau$ the requested accuracy. Other letters are local.

\subsection{Preliminaries and the GLS ellipsoid theorem}
\label{ell:sec:preliminaries}

\paragraph{Geometry and encoding.}
All balls are closed Euclidean balls:
\[
 B(a,\varrho):=\{x\in\R^p:\norm{x-a}_2\le \varrho\}.
\]
For a closed set $C$, write
\[
 C^{+\delta}:=\{y\in\R^p:\dist(y,C)\le\delta\},
 \qquad \varnothing^{+\delta}:=\varnothing.
\]
For a rational positive definite matrix $Q$ and a rational vector $z$, let
\[
 E(Q,z):=\{x:(x-z)^\top Q^{-1}(x-z)\le1\}.
\]
Volume means $p$-dimensional Lebesgue measure. All logarithms are to base two.
The notation $\bits{q}$ denotes the binary encoding length of a rational number
(numerator and positive denominator), with the corresponding sum convention
for vectors and matrices. Arithmetic is in the Turing bit model.

Let $\mathsf{len}$ denote the length of the data specifying the separation oracle for the
instance. An oracle-polynomial algorithm makes polynomially many oracle calls,
with polynomially bounded query lengths, and otherwise performs polynomially
many bit operations. If the oracle itself runs in time polynomial in $\mathsf{len}$ and
the query length, this gives an ordinary polynomial-time algorithm, with $\mathsf{len}$
included among the parameters.

We use strong rational separation as in Definition~\ref{def:strong}.

\paragraph{Numerical conventions.}
To make a running-time bound involving logarithms unambiguous, the main lemma
uses supplied bounds and accuracy of the form
\begin{equation}\label{ell:eq:dyadic}
 R=2^{k_1},\qquad \rho=2^{-k_2},\qquad
 \tau=2^{-k_3},\qquad k_1,k_2,k_3\in\mathbb Z_{\ge0},
\end{equation}
in ordinary binary rational encoding. Their encoding lengths are $O(1+k_1)$,
$O(1+k_2)$, and $O(1+k_3)$, respectively. This entails no loss of accuracy or
geometric validity. Given arbitrary positive rational inputs $R_0,\rho_0,
\tau_0$, replace them by
\begin{align*}
 R&=2^{\lceil\log\max\{1,R_0\}\rceil},\\
 \rho&=2^{\lfloor\log\min\{1,\rho_0\}\rfloor},\\
 \tau&=2^{\lfloor\log\min\{1,\tau_0\}\rfloor}.
\end{align*}
Thus $R\ge R_0$, $\rho\le\rho_0$, and $\tau\le\tau_0$.
These powers are found by integer comparisons and bit lengths, without
evaluating real logarithms. For arbitrarily encoded rational inputs, the time
to read and process those inputs must also be included. An arbitrary real
accuracy $\tau_0>0$ is handled by supplying a rational accuracy no larger
than it. In the main text $\tau_0=\varepsilon$ is the accuracy requested
in Theorem~\ref{thm:main}.

\begin{theorem}[Rational central-cut ellipsoid method; GLS]\label{ell:thm:GLS}
Let $p\ge2$, let $C\subseteq B(0,R)$ be closed and convex, and let $R\ge1$
and $0<\nu<1$ be rational. Suppose an oracle, given $y\in\Q^p$ and rational
$\delta>0$, either declares $y\in C^{+\delta}$ or returns $h\in\Q^p$ with
$\norm{h}_\infty=1$ satisfying
\begin{equation}\label{ell:eq:GLS-oracle}
 h^\top x\le h^\top y+\delta\qquad(x\in C).
\end{equation}
There is an oracle-polynomial rational algorithm returning either
\begin{enumerate}[label=\textup{(\roman*)}]
\item a rational point $z\in C^{+\nu}$ that is a queried ellipsoid centre
      accepted by the oracle; or
\item a rational ellipsoid $E(Q,z)$ containing $C$ with
      $\vol(E(Q,z))\le\nu$.
\end{enumerate}
The number of oracle calls is
\[
 O\!\left(p^2\log(2R)+p\log(1/\nu)\right),
\]
and query lengths and other bit costs are polynomial in
$p,\bits R,\bits\nu$ and the oracle input/output lengths.
\end{theorem}

\noindent
This is a restatement, in our notation and normalized parameter range, of
\cite[Theorem~3.2.1, p.~87]{GLS}. The iteration bound follows from
GLS~(3.2.2); its rational implementation and containment/volume analysis are
on pp.~87--93, especially Lemmas~3.2.8--3.2.10 on p.~88. The acceptance clause
in alternative~\textup{(i)} records the stopping rule on p.~88.
The restriction $p\ge2$ follows the dimension assumption that GLS make for
all proofs of their Chapter~3 (p.~66) and use in the proofs of
Lemmas~3.2.8 and~3.2.10 (pp.~89--90 and~93); the case $p=1$ is
handled explicitly below. No nonempty-interior hypothesis is needed here.

\begin{corollary}[Exact feasibility or a small enclosing ellipsoid]
\label{ell:cor:exact-or-small}
Let $p\ge2$ and let $R\ge1$ be rational. If a closed convex set
$C\subseteq B(0,R)$ has a strong rational separation oracle, then for every
rational $0<\nu<1$ an oracle-polynomial algorithm returns either a rational
point $z\in C$, or an ellipsoid containing $C$ of volume at most $\nu$.
Consequently, if $\vol(C)>\nu$, the algorithm must return a point in $C$.
\end{corollary}

\begin{proof}
Use the algorithm of Theorem~\ref{ell:thm:GLS}, with the following oracle adapter.
On a query $(y,\delta)$, call the strong oracle at $y$. Accept only if it
certifies $y\in C$. Otherwise normalize the returned vector by
$\widehat h=h/\norm h_\infty$. The strict separation inequality implies
\[
 \widehat h^\top x<\widehat h^\top y
 \le\widehat h^\top y+\delta\qquad(x\in C),
\]
so this is a valid oracle for Theorem~\ref{ell:thm:GLS}. Normalization is rational
and has polynomial bit cost. By the acceptance clause in Theorem~\ref{ell:thm:GLS}, every point returned
by the algorithm was accepted by the strong oracle and therefore belongs
\emph{exactly} to $C$. All rounding in the ellipsoid updates is already covered
by the GLS theorem; no rounding of the final point is performed.
If the ellipsoid alternative occurs, monotonicity of volume gives
$\vol(C)\le\nu$. This excludes that alternative when $\vol(C)>\nu$.
\end{proof}

\begin{remark}[The volume certificate]\label{ell:rem:not-empty}
The ellipsoid alternative certifies a volume upper bound, including for
nonempty sets with small volume. This is the certificate used by the
optimization proof below.
\end{remark}

\subsection{An explicit ball of nearly optimal points}
The separation-to-optimization arguments of GLS use lower volume
bounds to detect improving feasible regions; see the proof
of~\cite[Theorem~4.2.2, p.~106]{GLS}.
The next two elementary lemmas give the particular near-optimal-ball
and rational-volume estimates used here, starting from an inner ball
whose centre need not be supplied.

\begin{lemma}[Near-optimal ball]\label{ell:lem:cap-ball}
Suppose $K$ is compact and convex and
\[
 B(a,\rho)\subseteq K\subseteq B(0,R),\qquad
 R\ge1,\quad \rho>0.
\]
Let $c\in\R^p$, set $C_c:=\max\{1,\norm c_1\}$ and $M:=RC_c$, and let
$0<\eta\le1$. If $x^*$ maximizes $c^\top x$ over $K$, then
\begin{equation}\label{ell:eq:cap-ball}
 B\!\left((1-\alpha)x^*+\alpha a,\alpha\rho\right)
 \subseteq
 \{x\in K:c^\top x\ge\OPT-\eta/2\},
 \qquad \alpha:=\frac{\eta}{4M},
\end{equation}
where $\OPT=c^\top x^*$.
\end{lemma}

\begin{proof}
The maximum exists by compactness. Since $0<\alpha\le1/4$, convexity gives
\[
 (1-\alpha)x^*+\alpha B(a,\rho)\subseteq K.
\]
The set on the left is the ball in~\eqref{ell:eq:cap-ball}. Any of its points
has the form $z=(1-\alpha)x^*+\alpha w$, with $w\in B(a,\rho)\subseteq K$.
Both $x^*$ and $w$ belong to $B(0,R)$; hence
\[
 \OPT-c^\top z
 =\alpha c^\top(x^*-w)
 \le\alpha\norm c_2\norm{x^*-w}_2
 \le2\alpha RC_c=\eta/2.
\]
This proves the inclusion. Neither $a$ nor $x^*$ needs to be known to the
algorithm.
\end{proof}

\begin{lemma}[A rational volume threshold]\label{ell:lem:volume}
Let $p\ge1$ and $q>0$, and define
\[
 \nu(q):=\frac12\left(\frac{q}{p}\right)^p.
\]
Every convex set $C\subseteq\R^p$ containing a ball of radius $q$ satisfies
$\vol(C)>\nu(q)$. If $0<q\le1$, then $0<\nu(q)<1$.
Moreover,
\[
 \log(1/\nu(q))=1+p\log p+p\log(1/q).
\]
If $q$ is rational, then $\nu(q)\in\Q$ and
\begin{equation}\label{ell:eq:volume-bits}
 \bits{\nu(q)}=O\!\left(p(\bits q+\log(p+1))\right).
\end{equation}
\end{lemma}

\begin{proof}
The ball contains a cube of side length $2q/p$: the Euclidean distance from
its centre to a corner is $q/\sqrt p\le q$. Therefore
\[
 \vol(C)\ge(2q/p)^p>\tfrac12(q/p)^p.
\]
The range $0<\nu(q)<1$ for $0<q\le1$ and the logarithmic identity follow
from the definition. When $q$ is rational, rational exponentiation gives
the encoding estimate. This avoids computing the volume of the Euclidean
unit ball.
\end{proof}

\subsection{Exact-feasible approximate optimization}

The reduction from separation to approximate linear optimization over a
well-bounded convex body is classical; see
GLS~\cite[Theorem~3.2.1, Theorem~4.2.2, and Remark~4.2.5]{GLS}.
The interior-point result~\cite[Theorem~1.1]{deKlerkVallentin16},
recalled in~\cite[Theorem~15]{Gribling23}, also returns an exactly
feasible rational point. Its input includes a rational feasible centre
and inner and outer radius bounds in the affine space of the SDP.
Here we use the classical strong-separation approach, for which the
radius bounds suffice without supplying the centre.
Lemma~\ref{ell:lem:main} records the exact-feasibility version needed
here under a strong rational separation oracle, including the explicit
dependence on the inner radius and the requested accuracy.

\begin{lemma}[Exact-feasible approximate optimization]\label{ell:lem:main}
Let $p\ge1$, and let $R,\rho,\tau$ be supplied positive rational
numbers normalized as in~\eqref{ell:eq:dyadic}; in particular,
$R\ge1$ and $0<\rho,\tau\le1$.
Let $K\subseteq\R^p$ be compact and convex, with
\begin{equation}\label{ell:eq:main-geometry}
 B(a,\rho)\subseteq K\subseteq B(0,R).
\end{equation}
Assume that $K$ admits a polynomial-time strong rational separation oracle.
The centre $a$ need not be supplied. For every $c\in\Q^p$, the algorithm
described in the proof, using the rational central-cut ellipsoid method,
returns $y\in K\cap\Q^p$ such that
\begin{equation}\label{ell:eq:main-guarantee}
 c^\top y\ge\max_{x\in K}c^\top x-\tau.
\end{equation}
Its running time is polynomial in
\begin{equation}\label{ell:eq:main-time}
 p,\quad \mathsf{len},\quad 1+\log R,\quad 1+\log(1/\rho),\quad
 \bits c,\quad 1+\log(1/\tau).
\end{equation}
For arbitrary positive rational bounds and accuracy, first apply the
normalization in Section~\ref{ell:sec:preliminaries} and include the input
encoding lengths. The resulting point satisfies the originally requested
accuracy, and all feasibility assertions refer to the original set $K$.
\end{lemma}

\begin{proof}
First suppose $p\ge2$. The bounds and accuracy are normalized by hypothesis.
Denote the optimum by $\OPT$ and define the
following rational quantities:
\begin{equation}\label{ell:eq:parameters}
 C_c=\max\{1,\norm c_1\},\qquad
 M=RC_c,\qquad \eta=\frac\tau4,\qquad
 q=\frac{\eta\rho}{4M}=\frac{\tau\rho}{16M},\qquad
 \nu=\frac12\left(\frac{q}{p}\right)^p.
\end{equation}
Here $M\ge1$, $0<\eta\le1/4$, $0<q\le\rho\le1$, and $0<\nu<1$. Also
\begin{equation}\label{ell:eq:objective-bound}
 |c^\top x|\le\norm c_2\norm x_2\le M\qquad(x\in K).
\end{equation}

\paragraph{Step 1: obtain one exactly feasible point.}
Apply Corollary~\ref{ell:cor:exact-or-small} to $K$ with threshold $\nu$.
Since $K$ contains a ball of radius $\rho\ge q$,
Lemma~\ref{ell:lem:volume} gives $\vol(K)>\nu$. Thus the call returns a rational
point $y\in K$. If $c=0$, return this point immediately.
In what follows $c\ne0$.

\paragraph{Step 2: implement a threshold query with a controlled gap.}
For $\beta\in\Q$, let
\[
 K_\beta:=K\cap\{x:c^\top x\ge\beta\}.
\]
This is a closed convex subset of $B(0,R)$, possibly empty. A strong oracle
for it is constructed as follows. If a query $z$ satisfies $c^\top z<\beta$,
return $h=-c$ and $b=-\beta$, which satisfy
\[
 h^\top x\le b<h^\top z\qquad(x\in K_\beta).
\]
Otherwise call the oracle for $K$. Its separating inequality remains valid
for $K_\beta$, and a positive answer certifies $z\in K_\beta$.

Run Corollary~\ref{ell:cor:exact-or-small} on $K_\beta$ with the same threshold
$\nu$. The resulting threshold procedure has the following alternatives:
\begin{align}
 \textnormal{success: }& z\in K\cap\Q^p,\quad c^\top z\ge\beta;
 \label{ell:eq:threshold-success}\\
 \textnormal{failure: }& \OPT<\beta+\eta.
 \label{ell:eq:threshold-failure}
\end{align}
To prove the failure implication, suppose instead that
$\beta\le\OPT-\eta$. Lemma~\ref{ell:lem:cap-ball} gives a ball of radius $q$
whose objective values are all at least $\OPT-\eta/2\ge\beta$.
That ball is contained in $K_\beta$, so
$\vol(K_\beta)>\nu$ by Lemma~\ref{ell:lem:volume}. The enclosing-ellipsoid
alternative is then impossible. This proves~\eqref{ell:eq:threshold-failure}.
The procedure need not decide whether $K_\beta$ is empty.

\paragraph{Step 3: search while retaining an exactly feasible point.}
Initialize numerical bounds $\ell=-M$ and $u=M$, keeping the point $y$ found
in Step~1. Maintain the invariant
\begin{equation}\label{ell:eq:invariant}
 y\in K\cap\Q^p,\qquad
 \ell\le c^\top y\le\OPT\le u.
\end{equation}
It holds initially by~\eqref{ell:eq:objective-bound}. While $u-\ell>\tau$,
set $\beta=(\ell+u)/2$ and execute the threshold procedure. Update as follows:
\begin{enumerate}[label=\textup{(\alph*)}]
\item On success with point $z$, set $\ell\leftarrow\beta$ and
      $y\leftarrow z$, leaving $u$ unchanged.
\item On failure, set $u\leftarrow\beta+\eta$, leaving $\ell$ and $y$
      unchanged.
\end{enumerate}
Both updates preserve~\eqref{ell:eq:invariant}, by
\eqref{ell:eq:threshold-success} and~\eqref{ell:eq:threshold-failure}.
In a failure update, $\beta+\eta<u$ because
$u-\ell>\tau>2\eta$. Thus the upper endpoint strictly decreases.
If the current width is $w=u-\ell>\tau$, a success reduces it to
$w/2$, whereas a failure changes it to
\[
 w/2+\eta=w/2+\tau/4<3w/4.
\]
The loop therefore terminates after
\begin{equation}\label{ell:eq:search-count}
 O\!\left(1+\log(M/\tau)\right)
\end{equation}
threshold queries. At termination,~\eqref{ell:eq:invariant} yields
\[
 0\le\OPT-c^\top y\le u-\ell\le\tau,
\]
which is~\eqref{ell:eq:main-guarantee}. The point $y$ is exactly feasible because
it is always one of the points accepted by the strong oracle.

\paragraph{Step 4: bit complexity.}
The numbers in~\eqref{ell:eq:parameters} have polynomial encoding length in
the parameters of~\eqref{ell:eq:main-time}. In particular,
\begin{equation}\label{ell:eq:threshold-size}
 \log(1/\nu)
 =1+p\log p+p\log\!\left(\frac{16M}{\tau\rho}\right).
\end{equation}
Furthermore $\bits{C_c}$ is polynomial in $\bits c$ and
$\log M\le\log R+O(1+\bits c)$. Each ellipsoid call thus uses polynomially
many oracle queries of polynomial length, provided the threshold $\beta$
also has polynomial length. The following observation makes that last
point explicit and prevents a recursive bit-length bound.

Let $D_0$ be the product of the denominators of $M$ and $\eta$ in lowest
terms. After $j$ threshold queries, both $\ell$ and $u$ have denominator
dividing $2^jD_0$. Indeed, taking a midpoint introduces at most one additional
factor of two, and adding $\eta$ introduces no other denominator.
Both endpoints stay in $[-M,M]$, because the interval contracts at every
iteration. Hence their encoding lengths, and the length of the next
threshold, are
\[
 O\!\left(1+\log(M+1)+\log D_0+j\right).
\]
Together with~\eqref{ell:eq:search-count}, this is a uniform polynomial bound.
Notice that a successful update uses $\ell=\beta$, not $\ell=c^\top z$;
the returned point's coordinates never feed into the search thresholds.

The oracle for $K_\beta$ uses only a rational scalar-product comparison and
at most one call to the oracle for $K$. Its bit cost is polynomial in the
query length, $\mathsf{len}$, $\bits c$, and $\bits\beta$. Combining this fact,
Theorem~\ref{ell:thm:GLS},~\eqref{ell:eq:search-count}, and~\eqref{ell:eq:threshold-size}
proves the claimed total polynomial running time, including output length.
For example, the total number of calls to the oracle for $K$ is at most
\[
 O\!\left(
 \bigl(1+\log(M/\tau)\bigr)
 \bigl(p^2\log(2R)+p\log(1/\nu)\bigr)
 \right),
\]
with polynomial bit work per call.

\paragraph{Dimension one.}
If $p=1$, apply the preceding proof in dimension two to
$\overline K=K\times[-1,1]$ and objective $\overline c=(c,0)$.
The product contains $B((a,0),\rho)$ and is contained in $B(0,2R)$, because
$\rho\le1$ and $R\ge1$. Its strong oracle acts on a rational query
$z=(z_1,z_2)$ as follows. If $|z_2|>1$, return
\[
 h=(0,\operatorname{sgn}z_2),\qquad b=1.
\]
Then $h^\top x\le1<|z_2|=h^\top z$ for every $x\in\overline K$.
Otherwise query the oracle for $K$ at $z_1$: accept $z$ if it accepts $z_1$,
and lift a separating pair $(h_1,b)$ to $((h_1,0),b)$ if it rejects.
This is a polynomial-time strong rational separation oracle for $\overline K$.
Projecting the returned point onto
its first coordinate preserves exact feasibility and objective value.
The change of parameters is polynomial. This completes the proof for all
$p\ge1$.
\end{proof}

\paragraph{Application to the main text.}
The algorithm uses the numerical radius bound. The existential Boolean
simplex in Section~\ref{sec:primal} provides this bound directly from
the number of free coordinates.

\section{Bounded-certificate search and residual repair}
\label{app:search}
The passage from coefficient bounds to approximate SoS proof search uses
classical convex optimization; see the bit-complexity context in
Raghavendra--Weitz~\cite{RW} and the rounding and residual-repair
arguments in~\cite[Theorem~22 and Lemma~23]{Gribling23}.
Here we prove the particular reduction used
in the paper: a full-dimensional perturbation of the coefficient SDP,
followed by exact repair of a small residual using Boolean axioms.
A related exact-arithmetic technique is rational rounding followed by
projection onto the affine space of Gram matrices, under a strict
feasibility and error bound
\cite[Section~3.1, Proposition~8]{PeyrlParrilo08}.
The reduction below combines a perturbation of the coefficient SDP
with the degree-preserving Boolean correction of
Lemma~\ref{lem:residual}.
The proof invokes the GLS consequence established in
Appendix~\ref{app:ellipsoid} and the explicit correction in
Appendix~\ref{app:algebra}.

\begin{lemma}[Bounded-certificate search]\label{lem:certificate-search}
Let $t\ge2$ be even. Suppose a rational Boolean polynomial system admits
a real degree-$t$ Gram proof of a specified rational polynomial $f\ge0$
whose matrix entries and equality-multiplier coefficients have magnitude
at most a supplied rational $B\ge1$. For rational $\varepsilon>0$, an
algorithm finds a rational degree-$t$ Gram proof of $f+\varepsilon\ge0$
in time polynomial in $\binom{n+t}{t}$, the system and target encodings,
$\bits B$ and $\bits\varepsilon$.
\end{lemma}
\begin{proof}
List the independent symmetric Gram entries and equality-multiplier
coefficients as $z\in\R^T$, and write the coefficient equations of a
degree-$t$ proof as $Az=\mathbf f$. They include all ordinary monomials
of degree at most $t$, of which there are $M=\binom{n+t}{t}$.
The rational matrix $A$ is explicitly constructed by polynomial
multiplication. Its dimensions and encoding length are polynomial in
the stated input. Put $g_0=1$, and let $Q_i(z)$ denote its Gram blocks.
Use the retained degrees as in Section~\ref{sec:prelim}, and define the
known polynomial
\[
 G=\sum_i g_i\sum_{m\text{ indexing }Q_i}m^2,
 \quad U=\max\{1,M+\norm G_1\},\quad
 \eta=\min\{1,\varepsilon/(2U)\}.
\]
All these are rational and have polynomial encoding length. Set
$H=\max\{2,\max_j\sum_i|A_{ji}|\}$ and consider the compact convex set
\[
 C_\eta=\{z:|z_i|\le2B,\quad
    |(Az-\mathbf f)_j|\le\eta,\quad Q_i(z)+\eta I\succeq0\}.
\]
Let $z^0$ encode the assumed exact certificate. For $\norm h_2\le
\eta/(4H)$, each box coordinate of $z^0+h$ has magnitude less than
$2B$, each equation residual has magnitude at most $H\norm h_2\le\eta/4$,
and the perturbation in any symmetric Gram block has Frobenius norm at
most $2\norm h_2<\eta$. Thus
\[
 B(z^0,\eta/(4H))\subseteq C_\eta\subseteq B(0,2BT).
\]
The centre $z^0$ is not needed. Box and coefficient inequalities have
explicit rational separators. A failed PSD condition gives a rational
negative witness by Lemma~\ref{lem:psd}, and hence a strong separator
linear in $z$. A zero normal with a negative constant is impossible
because $C_\eta$ is nonempty. This is a polynomial-time strong rational
oracle. Normalizing the radius bounds to powers of two, apply
Lemma~\ref{ell:lem:main} with objective zero to obtain a rational
$z\in C_\eta$ exactly. All required logarithms and oracle encodings
are polynomial in the input parameters.

Replace each block $Q_i(z)$ by $\bar Q_i=Q_i(z)+\eta I\succeq0$, leaving
the equality multipliers unchanged, and let $P$ be the polynomial they
certify. It is an exact rational SoS proof of $P\ge0$. Before the shift
the coefficient error from $f$ has $\ell_1$ norm at most $M\eta$; the
shift adds $\eta G$. Hence for $e=f-P$,
\[
 \norm{\ml(e)}_1\le\norm e_1
 \le\eta(M+\norm G_1)\le\varepsilon/2.
\]
Lemma~\ref{lem:residual} gives a rational Boolean certificate of
$\varepsilon+e\ge0$ at degree $t$. Adding it to the certificate of
$P$ gives the exact polynomial identity for $f+\varepsilon$.
All final matrix and coefficient lengths are polynomial by the oracle
output bound and the explicit repair construction.
\end{proof}

\end{document}